\documentclass[a4paper, 11pt]{article}
\usepackage{fullpage}

\usepackage{amssymb,latexsym}
\usepackage{amsmath,amsthm,graphicx}
\usepackage{graphics}
\usepackage{graphicx}
\usepackage{color}
\usepackage{euscript}
\usepackage{geometry}
\usepackage{gensymb} %for degree - in the acknowledgement
\usepackage{hyperref}
\hypersetup{colorlinks=true,linkcolor=blue,citecolor=blue}
\usepackage[authoryear]{natbib}

\usepackage{dsfont} % for \mathds{1}
\usepackage{comment}

\usepackage{tikz}
\usepackage{ifthen}

\numberwithin{equation}{section}

\newtheorem{theorem}{Theorem}[section]

\newtheorem{lemma}{Lemma}[section]
\theoremstyle{definition}
\newtheorem{corollary}{Corollary}[section]
\newtheorem{definition}{Definition}[section]
\newtheorem{remark}{Remark}[section]
\newtheorem{example}{Example}[section]

\newtheorem*{assumption*}{Assumption}
\newtheorem{assumption}{Assumption}[section]

\definecolor{ForestGreen}{rgb}{.13,.54,.13}
\definecolor{BrickRed}{rgb}{.70,0,0}

\ifdefined\DRAFT
\newcommand{\fed}[1]{{\color{red}{(\textbf{Fedor:} #1)}}}
\newcommand{\ed}[1]{{\color{BrickRed} {#1}}}
\newcommand{\new}[1]{{{\color{ForestGreen} {#1}}}}

\else
\newcommand{\fed}[1]{}
\newcommand{\ed}[1]{#1}
\newcommand{\new}[1]{#1}

\fi

\ifdefined\DRAFT
\newcommand{\rann}[1]{\textcolor{red}{(\textbf{Rann:} #1)}}
\else
\newcommand{\rann}[1]{}
\fi

\newcommand{\E}{\mathbb{E}}

\renewcommand{\P}{\mathbb{P}}

\title{On social networks that support learning\thanks{We are grateful (in alphabetic order) to Herve Moulin, Alexander Nesterov, Matthew  Jackson, Nicolas Vieille, and Omer Tamuz for inspirational discussions and suggestions. We also thank attendees of the game-theory seminar at the Technion and the Conference on Mechanism and Institution Design 2020 for their feedback. \new{We are grateful to Michael Borns for proofreading the paper.}  \newline
		 Arieli's research is supported by the Ministry of Science and Technology grant \#2028255.\newline
		 Sandomirskiy's research is partially supported by the Lady Davis Foundation, 
		Grant 19-01-00762 of the Russian Foundation for Basic Research,  the European Research Council (ERC) under the European Union's Horizon 2020 research and innovation program ( \#740435), and the Basic Research Program of the National Research University Higher School of Economics. 		\newline
	Smorodinsky's research is supported by the United States-Israel Binational Science Foundation, National Science Foundation grant \#2016734, German-Israel Foundation grant I-1419-118.4/2017, Ministry of Science and Technology grant \#19400214, Technion VPR grants, and the Bernard M. Gordon Center for Systems Engineering at the Technion.
}}
\author{Itai Arieli \footnote{Technion IE\&M, Haifa (Israel).} \and Fedor Sandomirskiy $^\dagger{\;}$\footnote{Higher School of Economics, St.~Petersburg (Russia).} \and Rann Smorodinsky $^\dagger$}
\date{}

\begin{document}

\maketitle

\ifdefined\DRAFT
\begin{center}
	\color{red}{\textbf{A DRAFT. PLEASE, DON'T DISTRIBUTE}}
\end{center}
\fi

\begin{abstract}
	
It is well understood that the structure of a social network is critical to whether or not
agents can aggregate information correctly. In this paper, we study social networks that
support information aggregation when rational agents act sequentially and irrevocably.
Whether or not information is aggregated depends, inter alia, on the order in which
agents decide. Thus, to decouple the order and the topology, our model studies a random
arrival order.

\new{Unlike the case of a fixed arrival order, in our model the decision of an agent is unlikely to be affected by those who are far from him in the network. This observation allows us to identify a \emph{local learning requirement}, a natural condition on the agent's neighborhood that guarantees that this agent  makes the
correct decision (with high probability) no matter how well other agents perform. Roughly speaking, the agent should belong to a multitude
of mutually exclusive social circles.} %the agent should belong to a multitude
%of mutually exclusive social circles. In this case, roughly speaking, the agent can restrict
%attention to one friend from each social circle and rely on the majority action to be the
%correct action.
%\fed{I am not sure we should outline the condition here. Here it reads as an obvious observation, while it is not. I propose to erase ``In this case, roughly speaking, the agent can restrict
%	attention to one friend from each social circle and rely on the majority action to be the
%	correct action.''}

\new{We illustrate the power of the local learning requirement by constructing a family of social networks that guarantee information aggregation despite that} no agent is a social hub (in
other words, there are no opinion leaders). \new{Although, the common wisdom of the social learning literature suggests that information aggregation is very fragile, another application of the local learning requirement demonstrates the existence of networks where learning prevails even if a substantial  
 fraction of the agents are not involved in the learning process.}
 % (that
%is, others are indifferent between actions).
On a technical level, the networks we construct rely on the theory of expander graphs, i.e., highly connected sparse graphs
with a wide range of applications from pure mathematics to error-correcting
codes.

\end{abstract}

\section{Introduction}\label{sec_intro}
The way ideas and information propagate in society is critical for engineering political campaigns, evaluating new technologies, marketing products, and introducing  new social conventions. It is well known that whether or not the information is properly aggregated is highly dependent on the quality of information accessible to the individual agents, on the way information is transmitted from one agent to another, the order and frequency with which agents take actions, and, finally, on the topology of the underlying social network.

In this paper, we study the role that the topology of the social network plays in information aggregation among rational agents. We do so using a variant of the herding model of \citet{banerjee1992simple} and \citet{bikhchandani1992} adapted to a social network setting. Agents decide sequentially on a binary action based on a private (bounded) signal and history of actions taken by their predecessors. In the variant we study, the set of predecessors an agent observes is restricted \new{by his set of neighbors in an exogenously given social network.}

Since agents act sequentially a network may properly aggregate information for a particular sequence while failing to do so for most other sequences. \ed{To decouple the network's topology from the order in which agents take their actions, we make a distinction between the \emph{a-priori network}, which is exogenously given, and the \emph{realized network}, which is induced by the aforementioned social network and the ordering of agents.
Such a distinction between the social network and realized observability structure is not standard.}
Our goal in this paper is to find a sufficient condition \ed{on the a-priori network that} guarantees learning for most sequences. We say that such a network {\em supports learning}. Our first step in the analysis, which is of independent interest, is to focus on a particular agent in the \ed{a-priori} network and characterize features of the {\bf local} structure of the network that guarantee that this agent will make the correct decision (with high probability). \new{Indeed, we show that an agent that 
belongs to a variety of  social circles that, among themselves, are mutually exclusive and socially distant has a high probability of taking the correct action.} 
%\fed{The essence of the condition is that the agent must be bridging clubs of agents that don't interact much with each other. Namely $v$ must be able to spot a large enough group of his friends having high enough degree and such that, the closest common friend except $v$ is far enough for any pair of them. The condition of bridging those agents who otherwise need many handshakes to connect each other allows to exclude herding among them.}
We refer to this condition as the \ed{\emph{local learning requirement.}} Therefore, social networks with the property that each agent satisfies the local learning requirement will aggregate information. We show, with an example, that this condition is not vacuous. Quite surprisingly, we demonstrate the existence of {\bf symmetric} networks with the aforementioned property. 

In the sociology literature, and in particular the literature on mass communication, it is often argued that learning is facilitated vis-\`{a}-vis \new{a small number of opinion leaders predetermined by their position in the social network}~\citep{katz1955personal}. Thus, the existence of  symmetric networks that support learning is quite counterintuitive as it implies that learning obtains without opinion leaders \new{(see the discussion in Section \ref{sec_discuss}).}

When studying learning in social networks it may be quite unrealistic to assume that 
every decision problem pertains to all agents in the society. In fact, in large societies (such as those present in online social networks) it is far more realistic that only a minority of the agents have a stake in an arbitrary decision problem. For example, a dilemma between two candidates for mayorship may interest one subset of the population while a dilemma between two competing ``green'' technologies (e.g., hybrid engines vs. electric engines) may be relevant to some other small set of agents. Thus, the proper requirement \ed{on a social network} is that information be properly aggregated when only a fraction of agents participate. We say that a network  {\em supports robust learning} if any subnetwork containing a constant fraction of the agents supports learning.\footnote{An alternative, weaker definition could require that most such subnetworks support learning. We discuss this \new{alternative definition} in Appendix~\ref{app_randomized_robustness}.} 
To obtain robust learning it is sufficient that in any network induced by a fraction of the agents most of those agents satisfy the local learning requirement. We show, with an example, that this condition is not vacuous. Once again, the example we have satisfies symmetry assumptions on the agents.

%In the spirit of the current public discourse on disinformation over  popular social networks, we propose another criterion for information aggregation. Assume a minority of the agents is controlled by a malicious player who dictates their action. The information of these agents is obviously lost and, furthermore, these agents may adversely effect the information aggregation of the rest of society. We will say that a social network is {\em immune to fake news} if information is aggregated even if a minority of the agents is controlled by a malicious player whose goal is to maximize the number of agents taking the wrong actions. It turns out that such networks also exist. In fact, the example we construct for a network that supports robust learning is also immune to fake news.

To construct social networks that support (robust) learning we tap into the theory of expander graphs and in particular  we make use of well-known properties of a family of graphs known as {\em Ramanujan} expanders.

\subsection{Related literature}
\new{The literature on information aggregation studies many aspects of social learning. In our literature review, we focus on the connection between the topology of the social network  and information aggregation. This connection} has received attention in the literature on repeated interaction in social networks, \new{where, in contrast  to our model, agents can revise their decision as many times as they want.}
For example, \citet{mossel2015strategic}  study Bayesian learning by rational agents and provide a sufficient condition on the network structure such that asymptotically agents select the optimal action with
a probability that approaches one as the network grows. \citet{golub2010naive} analyze conditions under which a ``naive'' updating process converges to a
fully rational limiting belief, which equals the Bayesian posterior conditional on all agents’ information, in a large society. The conditions require no agent to be too influential.
%[NEED A PARAGRAPH ON THIS TOPIC - maybe you can use this text:
%Some formal models have also been used to study information dissemination over networks. Much of this literature, starting with \citep{degroot1974reaching}, focuses on boundedly rational agents, {e.g.~\citep{golub2010naive,ellison1993rules}.}]

In the context of sequential social \new{learning, where each agent takes an action only once,} it was first suggested by  \citet{smith1991essays} that learning can be obtained by restricting the social connections among agents. In particular, Smith points out
%makes the trivial  observation 
that, for a given sequence of agents, if the deciding agents initially do not observe each other then they collectively form a sample of independent observations (\new{sometimes referred to as} the ``guinea pigs''). Thus, any subsequent agent observing them is likely to make the correct decision. \new{\citet{sgroi2002} discusses how to choose  the set of guinea pigs optimally.  Unfortunately, these insights couple} the network structure with the order in which agents make their decision and is, therefore, moot to the questions we pose.
%\citet{mossel2019social} study more comprehensively the  interplay between the
%structure of the network and the structure of information for the agents that is required to achieve learning. 
%\fed{ITAI, do you agree that the paragraph proposed above better reflects the paper by \citet{mossel2019social}? }
%\fed{I wonder if \citet{mossel2019social} is a wrong reference here: they (almost) don't discuss the network structure... Although, if it is the right reference I propose to address this paper as follows:
%``\citet{mossel2019social} develop a general framework that captures the steady-state of various social learning models with a large number of agents. Their approach is especially effective to analyze how the information structure affects learning in the case of repetitive actions and unbounded signals, whereas we are interested in the the way learning is affected by network's topology assuming that agents act sequentially and fixing a particular simple information structure.''
%I did not mention here random arrival as a distinction since I believe that it can be captured within their framework too even though they do not discuss it.}

\ed{The interplay between the network structure and sequential learning is studied by~\citet{Acemoglu} in a model with  rational agents and sequential actions.} \new{In their model, prior to making his decision, an agent observes a random sample of predecessors, and the paper focuses on sufficient conditions on this sampling distribution that guarantee learning. The network} structure \new{turns out to be} coupled with the order in which agents make their decisions. Although
the \new{sampling model of~\citet{Acemoglu}} induces a random realized network, this network cannot be derived from and so is not associated  with an exogenously given \ed{a-priori} network.\footnote{%One natural question when comparing the two approaches is {whether} the model of a random directed {network} {by \cite{Acemoglu}} can be {approximately} replicated in our model with a fixed {network} and a random order. We do not have a definite answer to this question. However, we suspect this is not the case since, 
	Unlike the independence of the random observation sampling assumed by \citet{Acemoglu}, in our model the random order generates a high correlation between the observation {sets} of individuals. Indeed, if an individual observes only a small subset of his friends when he makes his decision he may infer that those friends are early arrivals and, therefore, based their own decisions on very limited information. This is not the case in the random sample model.}
 Thus, our approach complements that of \citet{Acemoglu} in the sense that their model better suits settings where the network structure is generated ad hoc  whereas our model fits social networks whose {a-priori} structure is fixed with no connection to the order in which agents decide in one problem or another and so the same network underlies a variety of decision problems, each with its own order.
%\citet{Acemoglu} study sequential decision by rational agents in a social network. In their model the network structure is coupled with the order at which agents take their decision. More explicitly, each agent, upon making his decision, get to observe a random sample of the agents that have already made their decision. They show that if most agents observe all the predecessors while a decreasing fraction of them observes none then information is aggregated. This results sheds no light on the required structure of a given network.%
\citet{Arieli2016} also \ed{consider} random sampling but assume the edges in the {network} are limited to the $m$-dimensional lattice.

\cite{bahar2019celeb} are the first to study social learning over an exogenously given network with a random arrival order of the agents. They demonstrate the existence of a social network where learning is guaranteed. Their network has a particular structure. It is a bipartite graph with a minority of agents on one side (whom they refer to as ``celebrities'') and a majority on the other side (``commoners''); \ed{for more details see Example~\ref{ex_celebrity}.} {Although the family of  celebrity graphs supports learning, they are sensitive to the participation of agents. In other words, in decision problems that are \ed{irrelevant} to a minority of the agents, information need not be properly aggregated and social learning may fail; i.e., \ed{the celebrity graphs} do not support robust learning}.

Recently, various papers  have demonstrated how fragile social learning can be (see, e.g., \citet{bohren2016informational}; \cite{frick2019misinterpreting}; \citet{mueller2018manipulating}). In contrast to these findings we {construct} a network structure that is robust in the sense that learning prevails even if a majority of agents, chosen adversarially, does not participate.

On a more technical level, the {networks}  we construct rely on the theory of expander graphs, i.e., highly connected sparse graphs with a wide range of applications from solving problems in pure mathematics to designing error-correcting codes (see~\citet{lubotzky2012expander} for a survey). Recently, this theory was  applied to problems of social learning by~\citet{mossel2014majority} and~\citet{feldman2014reaching} in the context of boundedly rational agents who take actions repeatedly.

\subsection{Structure of the paper}  
{Section~\ref{sec_model} contains the description of the model. In Section~\ref{sec_local_structure}, we demonstrate that whether or not an agent learns is determined by the local structure of the a-priori network around him and derive the local learning requirement on this neighborhood, which guarantees that the agent learns the state.
Sections~\ref{sec_egalitarian} and Section~\ref{sec_robust} are devoted to implications of the local learning requirement.  Section~\ref{sec_egalitarian} shows that learning is possible in egalitarian societies: there are symmetric networks that support learning.  Section~\ref{sec_robust} strengthens this result by constructing symmetric networks where learning is robust to adversarial elimination of groups of agents. 
Section~\ref{sec_discuss} discusses alternative interpretations and extensions of the basic model.}% and Section~\ref{sec_conclusions} concludes.

{The results of Sections~\ref{sec_egalitarian} and~\ref{sec_robust}
heavily rely on the insights from the theory of expander graphs, which are explained in Appendix~\ref{sect_expanders}. Technical proofs for Sections~\ref{sec_local_structure} and~\ref{sec_robust} are in Appendices~\ref{app_localization}, \ref{appendix_local_requirement}, and~\ref{appendix_missed_proofs_robust}. Appendix~\ref{app_randomized_robustness} discusses an alternative, weaker notion of robustness, where groups of agents are eliminated at random.}

\section{The model}\label{sec_model}

%We make three small modifications in the original model of~\cite{bahar2019celeb}:
%``friendship'' is always mutual as on Facebook, signals are binary, and, in additional to observing their arrived friends and their actions, agents also learn who of their second-order friends (i.e., friends of friends) have already arrived. These assumption are made to simplify the statements and can be partially relaxed, see Subsection~\ref{subsect_extensions}.
A social network is an undirected graph $G=(V,E)$, where $V$ is the set of agents and an edge $vu$ is contained in $E$ if $v$ and $u$ are ``friends.'' {Friendship is always mutual, as on Facebook: $vu\in E\Rightarrow uv\in E$.} 
We denote by $F_v$ the set of $v$'s friends $\{u\in V\, : \, vu\in E \}$. 

With each agent $v\in V$, we associate his arrival time $t_v$; the arrival times {$T=(t_v)_{v\in V}$} are independent random variables uniformly distributed on\footnote{Random arrival order allows us to disentangle the topology {of the a-priori network} and the order in which agents make their decisions. Randomness is also justified by the fact that this order usually differs from one issue to another (e.g, iPhone vs. Android, public kindergartens vs. private ones). \ed{We stress} that \ed{in our model} {\em arrival} refers to the time \ed{when} the agent makes a decision, and not to the time when \ed{he} joins the network \ed{as in the random sampling model of}~\cite{Acemoglu}.} $[0,1]$.
The collection of arrival times $T$ induces the orientation on edges of $G=(V,E)$ from late arrivals to early ones and converts it into the direct network $G_T=(V,E_T)$ with $E_T=\{vu\in E\,:\, t_u<t_v\}$, {which we call \emph{the realized network}.} \ed{We will refer to $G$ as the \emph{a-priori network} in order to distinguish it from $G_T$.}
{The directed edge $vu$ \ed{in the realized network} means that agent $v$ gets to observe agent $u$; i.e., $v$ observes his friends who arrived earlier. We denote the set of all such friends of $v$ by $F_{v,T}=\{u\in V\, :  \, vu \in E_T\}\subset F_v$.}

Upon arrival, each agent $v$ takes an action $a_v\in\{0,1\}$, depending on the information available. The agent gets a payoff of one if $a_v=\theta$, where $\theta$ is the random state, a Bernoulli random variable with success probability $1/2$; {if the action does not match the state, the agent gets zero payoff.}  Nobody observes $\theta$ but every agent receives a binary signal $s_v\in\{0,1\}$ that equals $\theta$ with probability {$p$} and $1-\theta$ with probability $1-p$, where {$\frac{1}{2}<p\leq 1$.} Signals are independent conditional on $\theta$. In addition to his own signal $s_v$, each agent $v$ observes the set of his friends who arrived earlier, $F_{v,T}$, and their actions, $(a_u)_{u\in F_{v,T}}$. We denote the information set of an agent $v$ by $I_v=(s_v, F_{v,T}, (a_u)_{u\in F_{v,T}})$. {Note that an agent  knows neither his arrival time nor the set of agents who arrived before him (except for his friends).}

All agents are rational and risk-neutral; the description of the model, probability distributions, and the {a-priori network} $G$ are common knowledge. {By contrast, the realized network is not observed by agents.}

%{ the set of all agents who arrived earlier $A_{<v}=\{u\in V\, : \, t_u<t_v \}$ and actions of his arrived friends $F_{v,T}=\{u\in V\, :  \, vu \in E,\  t_u<t_v\}\subset A_{<v}$.}\footnote{{Binary signals, {and observation of the whole set of arrived agents are technical assumptions. The first one can be easily relaxed and the second is only used in Section~\ref{sec_random_subsets}, see also the discussion in Subsection~\ref{subsect_extensions}.}}}
%\fed{Observation  of second-order friends turns out to be crucial in Subsection~5.1}
%\vskip -5cm

\subsection{Equilibria}\label{subsubsect_equilibria}
A mixed strategy $\sigma_v$ of an agent $v$  maps his information set $I_v$  to the probability distribution on $\{0,1\}$, according to which his action $a_v$ is then chosen.  {The goal of each agent is to maximize his expected payoff, which coincides with the probability of taking the action that matches the state.} %Note that according to this definition, an agent know neither his arrival time nor the set of agents arrived before him (except his friends).

{We consider an equilibrium $\sigma=(\sigma_v)_{v\in V}$ of the induced Bayesian game and use $\P_\sigma$ and $\E_\sigma$ for the probability and for the expectation with respect to all the randomness in the problem (the state, signals, arrival times, and actions). An equilibrium exists since it is a finite game; however, it  may be non-unique. {We omit the subscript $\sigma$ and write $\P$ and $\E$ when this creates no confusion.}}

%\fed{State-symmetric equilibria are defined formally. The discussion of tie-braking, multiplicity and (lack of) sequential structure is expanded}

{We say that an equilibrium is \emph{state-symmetric} if the distribution $\P$ is invariant under the mapping $\big(\theta,\, (s_v)_{v\in V}, \, (a_v)_{v\in V}\big)\to\big(1-\theta,\, (1-s_v)_{v\in V}, \, (1-a_v)_{v\in V}\big)$, i.e., under simultaneous flipping of the state, signals, and actions. Since the payoffs enjoy this symmetry, a state-symmetric  equilibrium exists.}

{In any equilibrium, an agent $v$ selects $a_v=1$ if $\P(\theta=1\mid I_v)> \frac{1}{2}$, i.e., if conditionally on the information, the state~$\theta=1$ is more likely. Similarly, $a_v=0$ if $\P(\theta=1\mid I_v)< \frac{1}{2}$. It may look as if equilibria can only differ in tie-breaking; however, it is not the whole truth since $\P(\theta=1\mid I_v)$ depends on the equilibrium strategies of other agents, which, in turn, are optimal replies to the strategies of others, including $v$.
{In particular,} an equilibrium lacks the sequential structure since no pair of agents $v,u\in V$ such that $vu\notin E$  know which of them acted first.}

%In other words, equilibria can only differ in tie-braking: how $v$ behaves when $\P(\theta=1\mid I_v)=\P(\theta=0\mid I_v)=\frac{1}{2}$, i.e, when $v$ is indifferent between the two actions.
%\fed{Actually, is existence of equilibrium clear for any tie-breaking? I worry about the claim ``equilibria can only differ in tie-braking'': is it true at all? Is  equilibrium unique given the tie-breaking?}

\subsection{Learning}
%\fed{I introduced the new notation: ``learning quality $L$''. This allows to compress and make more formal many of the arguments where we discuss ``average number of correct actions''}
{The following definition quantifies how well an agent and the a-priori network itself aggregate information.}
\begin{definition}[Learning quality]
{For an {a-priori} network $G=(V,E)$ and an equilibrium $\sigma$,
the \emph{learning quality}  of an agent $v\in V$ is the probability that this agent takes the correct action:
$$l_{\sigma}(v)=\P_\sigma (a_v=\theta).$$
The \emph{learning quality} of the network $G$ is the expected number  of agents taking the correct action: 
$$L_{\sigma}(G)=\frac{1}{|V|}\cdot \E_\sigma \big|\{v\in V\,:\, a_v=\theta \}\big|=\frac{1}{|V|}\sum_{v\in V} l_{\sigma}(v).%=\frac{1}{|V|}\cdot \E_\sigma \big|\{v\in V\,:\, a_v=\theta \}\big|.$$
$$
We say that an agent with $l_{\sigma}(v)$ close to one \emph{learns the state} and the network with $L_{\sigma}(G)$ close to one \emph{supports learning}. The following example shows that networks may fail to support learning due to the herding phenomenon.}
\end{definition}
%We denote the learning quality for the best equilibrium $\sigma$  by  $L(G)=\max_{\sigma} L_{\sigma}(G)$. The maximum is attained since the set of equilibria in a finite game is closed.
%\begin{definition}  We say that a sequence of networks $G_n=(V_n,E_n)$ \emph{admits asymptotic learning} {if   $L(G_n)$
%tends to one, as $n$ grows.}
%\end{definition}
%\fed{The definition of learning is weakened: now we have ``for some sequence of equilibria'' instead of ``for all sequences''. This allows discussing an alternative notion in Section 4 and talk about learning for symmetric equilibria in Sections 5-6.}
\begin{example}[Herding spoils learning]\label{ex_clique}
Let the a-priori network $G$ be the $n$-clique $K_n$, the complete graph on $n$ vertices. In this case, every agent $v$ observes actions of all those who came earlier. {We consider an equilibrium $\sigma$, where in the case of indifference, each agent follows his signal (one can show that learning quality in other equilibria can only be worse).}
If the first two agents take the incorrect action $a=1-\theta$, then the third agent will ignore his signal repeating this wrong action since the chance that both his predecessors are wrong is lower than the chance of him getting the wrong signal, and so on.
%\footnote{The fact that the first two agents are informative does not satisfy in any equilibrium, nevertheless, it is easy to see that $L(G_n)\leq 1-(1-\varepsilon)^2$ holds true in every equilibrium.} 
This phenomenon known as herding ruins information aggregation. Indeed,  with a probability of at least $(1-p)^2$ all the agents take the incorrect action. Thus  {$L_\sigma(K_n)\leq 1-(1-p)^2$;} i.e., the learning quality is bounded away from $1$ even for large cliques. {By symmetry, $l_\sigma(v)=L_\sigma(K_n)$ for all agents $v$.}
%there is no asymptotic learning for the sequence $G_n$.
\end{example}
\citet{bahar2019celeb} {consider a similar model of learning with random arrivals  and ask whether there exist networks that support learning.} They provide an affirmative answer to this question by identifying a family of \emph{celebrity graphs}, \new{the only known family of networks that supports learning with random arrivals.} We will construct another such family in Section~\ref{sec_egalitarian}. 

\begin{example}[Celebrity graphs support learning]\label{ex_celebrity}
Consider a two-tier society, where a large set of $k$ ``commoners'' observe  a large but smaller set of $m$ ``celebrities,''  $1\ll m\ll k$. The corresponding {a-priori}  network is a complete bipartite graph $B_{k,m}$; {see Figure~\ref{fig_celebrity}.}

{On average,  a set of $\frac{k}{m+1}\gg 1$ commoners  arrive  before the first celebrity. These commoners take their actions in isolation, i.e., without observing anybody else, and so these actions are dictated solely by their signals. Such isolated commoners play the role of ``guinea pigs'' \citep{sgroi2002}. 
Since $\frac{k}{m+1}\gg 1$, the law of large numbers suggests that the first celebrity aggregates information from these i.i.d. inputs and thereby takes the correct action with high probability. Subsequent commoners observe this celebrity and, hence, are also likely to make the right choice. This correct action propagates  and so we conclude that the whole population except for a negligible fraction of $\frac{1}{m+1}$ commoners takes the right action with high probability.} 
This informal reasoning suggests that for any $\delta>0$ we can find  $k$ and $m$ such that the learning quality of {each agent and of the network itself} is at least $1-\delta$.
The formal argument in~\citet{bahar2019celeb} is tricky since a celebrity must ``guess'' which observed commoners are isolated and which are probably not.

{Although the celebrity graphs support learning they are not robust to elimination of small groups of agents. Indeed, the elimination of all the celebrities  would render the commoners completely isolated and so learning would not obtain. \new{We also note that the fact that celebrity graphs support learning hinges on the assumption of a random arrival order. Indeed, if all commoners arrive before celebrities, then all of them make a decision in isolation and learning fails.}

} 

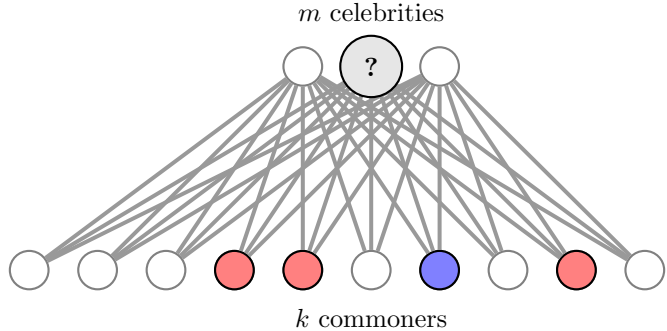
\begin{figure}[h]
	\definecolor{olive}{cmyk}{0.21,0,0.56,0.58}
	\begin{center}	
		\begin{tikzpicture}[transform shape,line width=1.5pt,scale=0.9]
\foreach \x  in {1,2,3}{%
	\node[draw,inner sep=0.2cm, circle , black!50] (C-\x) at (4+\x, 1.5) [thick] {};
}

\foreach \y  in {1,...,10}{%
	\node[draw,inner sep=0.2cm, circle , black!50] (N-\y) at (\y, -1.5) [thick] {};
}

	\node[draw,inner sep=0.2cm, circle , fill=red!50] (N-9) at (9, -1.5) [thick] {};
	\node[draw,inner sep=0.2cm, circle , fill=red!50] (N-5) at (5, -1.5) [thick] {};
	\node[draw,inner sep=0.2cm, circle , fill=blue!50] (N-7) at (7, -1.5) [thick] {};
	\node[draw,inner sep=0.2cm, circle , fill=red!50] (N-4) at (4, -1.5) [thick] {};

%
%\foreach \y  in {\only<7->{7}}{%
%	\node[draw,inner sep=0.2cm, circle , fill=blue] (N-\y) at (1+\y, -1.5) [thick] {};
%}
%
%\foreach \y  in {\only<3->{9},\only<5->{5},\only<9->{4}}{%
%	\node[draw,inner sep=0.2cm, circle , fill=red] (N-\y) at (1+\y, -1.5) [thick] {};
%}

\foreach \x  in {1,...,3}{%
	\foreach \y in {1,...,10}{%
		\path[black!40] (C-\x) edge[-] (N-\y);
	}
}

\node[draw,inner sep=0.2cm, circle , fill=black!10] (C-2) at (4+2, 1.5) [thick] {\textbf{?}};

\node at (6, 2.3) {{{$m$ celebrities}}};
\node at (6, -2.2) {{{$k$ commoners}}};
		\end{tikzpicture}
	\end{center}
	\caption{The celebrity graph. When the first celebrity arrives, he typically observes $\frac{k}{m+1}\gg1$ commoners who made their decisions in isolation. These decisions match the signals and, hence, the celebrity learns from a number of independent sources. As a result, he  takes a well-informed action, and then this action propagates.
	\label{fig_celebrity}
	}
\end{figure}
\end{example}

%\begin{example}[Random arrival order is crucial]\label{ex_random_order}
%One important modeling aspect is that the arrival order of agents is random.
%For cliques from Example~\ref{ex_clique}, the assumption of random arrival order is inconsequential. However, for arbitrary networks,  such as celebrities graphs of Example~\ref{ex_celebrity}, this assumption is critical. To see this  one can consider a celebrities graph, where all commoners arrive before celebrities and thus all of them make a decision in isolation and learning fails. 
%Alternatively, consider a   network $G_n$ that is almost $n$-clique, except that $k_n$  agents have been singled out and do not observe each other. Assume that $k_n$ tends to infinity but $\frac{k_n}{n}\to 0$.  Then, in the particular arrival order where these $k_n$ arrive first, the sequence $G_n$ demonstrates asymptotic learning. However, if the arrival order is random then with high probability we will observe the herding phenomenon.
%\end{example}

{
\subsection{Notation} Throughout the paper we use the following  notation. For a real number $x$, the floor and the ceiling  are denoted by $\lfloor x \rfloor$ and $\lceil x\rceil$, respectively. The former is the biggest integer $n$ such that $n \leq x$ and the latter is the smallest integer $n\geq x$.  The base of the natural logarithm is denoted by $e$.

 Consider a network $G=(V,E)$ that is possibly directed. By $\deg_G (v)$ we denote the total degree of a vertex $v$, i.e., the number of $u\in V$ such that at least one of the edges, $uv$ or $vu$, is in $E$. The network is \emph{$D$-regular} if $\deg_G(v)=D$ for all $v\in V$.   A map $f: \ V\to V$ is called an \emph{automorphism} of $G$ if $f$ is a bijection and it preserves edges, i.e., $uv\in E \Leftrightarrow f(u)f(v)\in E$.
 
 \new{A network $G'=(V',E')$ is a \emph{subnetwork} of $G$ (denoted by $G'\subset G$) if $V'\subset V$ and $E'\subset E\cap (V'\times V')$. The subnetwork is \emph{induced} by the set of vertices $V'$ if $E'=E\cap (V'\times V')$; such a subnetwork is denoted by $G^{V'}$.}
% Given a subset $V'\subset V$, the \emph{induced subgraph} is denoted by $G^{V'}$, i.e.,  $G^{V'} =(V',E')$, where $E'=\{uv \in E\,:\, u,v\in V'\}$. 
\new{We denote by $G\setminus v=G^{V\setminus \{v\}}$ the induced subnetwork obtained by deletion of a given vertex~$v$.}

A \emph{path of length $k$} in   $G$ is a sequence of vertices $(u_0,u_1,\ldots, u_k)$ such that
all the edges $u_i u_{i+1}$, $i=0,1,\ldots,k-1$ belong to $E$.
 The distance $d_G(v,v')$ between  two vertices $v,v'\in V$ is the minimal $k$ such that there is a path of length $k$ with $u_0=v$ and $u_k=v'$; if there is no such $k$, the distance is infinite, $d_G(v,v')=+\infty$.  The \emph{$r$-neighborhood} $B_{r,G}(v)$ of a vertex $v$ is the set of all vertices $v'$ such that $d_G(v,v')\leq r$; the number $r$ is the \emph{radius} of the neighborhood. \new{Abusing the notation, we will not distinguish between the $r$-neighborhood $B_{r,G}(v)$ (the set of vertices) and the subnetwork of $G$ induced by this set of vertices.}
 %	 the induced subgraph of $G$ with the set of vertices use the same notation $B_{r,G}(v)$ for  with the set of vertices 	denote the induced subgraph of $G$ with  }
A \emph{cycle of length $k$} is a path with $u_0=u_k$ such that all vertices $u_0,u_1,\ldots, u_{k-1}$ are distinct (sometimes cycles with no repetitions are called simple). The \emph{girth} $g_G$ is the length of the shortest cycle.

\new{When referring to the a-priori network $G$, we will omit the dependence of all the objects on $G$ and write simply $\deg(v)$, $d(v,v')$, $B_r(v)$, and $g$.}  
}

%; abusing the notation, we also refer to the subnetwork of $\Gamma$ spanned by $B_r,\Gamma(v)$ as the $r$-neighborhood of $v$.

	%A \emph{cycle of length $k$} is a path with $v_0=v_k$ such that all vertices $v_0,v_1,\ldots, v_{k-1}$ are distinct (sometimes cycles with no repetitions are called simple). A graph is a tree if it has no cycles.}

%\subsection{Our question}
%{The theory of two-step information flow and Example~\ref{ex_celebrity} suggest that social learning has two properties:}
%\begin{itemize}
%	\item Most of agents learn the state ``second-hand''. {They do not observe content-generators {(isolated agents)} directly but observe intermediaries (celebrities in Example~\ref{ex_celebrity}).}%, who observe many content-generators.
%	\item {Learning is facilitated by a  minority of agents (e.g., celebrities in Example~\ref{ex_celebrity} or influentials in the context of the two-step information flow model). Once these agents are deleted from the network, learning does not persist.}
%\end{itemize}
%Our goal  is to understand whether these  two properties are indeed immanent features of social learning.

{
\section{Learning and the local network structure}\label{sec_local_structure}
In this section we study the connection between the local structure \new{of the a-priori network in the neighborhood of an agent} and his learning \new{quality.} 

%e.g., in the standard herding model of~\citet{banerjee1992simple} and \citet{bikhchandani1992}, the action of the first agent

%A major impediment for learning in the standard herding model is that if the first two agents take a wrong action, then all other agents, even those which arrive much later, take the wrong action. Thus, intuitively speaking, a necessary condition for learning to hold is that agents which are far away rarely affect each other's action. Next, we study the connection between the distance between agents and the probability that they affect each other actions.  

%Recall the distinction between the a-priori network $G$ that represents the structure of social connections, and the realized network $G_T$ that captures possible paths of information transmission from one agent to another for a given assignment of arrival times $T$. %In particular, the action $a_u$ made by an agent $u$ can affect the decision of agent $v$ only if there is a directed path from $v$ to $u$ in $G_T$.

\new{In the case of a fixed arrival order, a major impediment for learning is that an agent can affect the decision of those who are far from him in the social network, which results in the possibility of large information cascades involving most of the network. Surprisingly, in our model with  random arrival order,
the decisions have a \emph{local nature}:} in Section~\ref{subsect_local_nature}, we show that if a pair of agents are far from each other in the a-priori network $G$, \ed{with high probability %there is no directed path between them in $G_T$, i.e., 
the action of one cannot affect the other.} 
This observation suggests that the learning quality of an agent must be determined by the local structure of his neighborhood in the a-priori network $G$, not by the global topology of $G$.

%\footnote{{This observation is not presented in the rest of the social learning literature, which does not distinguish a-priori and realized networks.}}} 

%\fed{
%TO BE ADDRESSED BY ITAI
%``Perhpas we should contrast this with the classical result that says that remore agents can adevrsally effect each other. We can give an example  to this sipirit for a given order. A mistake by the first two agents propagate to the 1000th agent. Hiwver, when the order is random this is not longer true.''}
%

{\new{The local nature of decisions}
	 motivates one to look for a condition on the topology of an agent's neighborhood that ensures that this particular agent learns the state with high probability no matter what the global topology of the network is and the learning qualities of other agents are. We identify such a \emph{local learning requirement} in Section~\ref{subsec_local_requirement}.}

{
\subsection{\new{Local nature of decisions}}\label{subsect_local_nature}
For a pair of agents $v$ and $u$, the action of $v$ may be affected by the choice made by $u$  if $v$ observes $u$, or if $v$ observes somebody who observes $u$, or if $v$ observes somebody who observes somebody who observes $u$, and so on. 

Given the collection of arrival times $T$, the agent $v$ can be influenced by $u$ only if there is a path $(v=u_0,u_1,u_2,\ldots,u_k=u)$ in the a-priori network $G$ such that $t_{u_i}>t_{u_{i+1}}$ for all $i$ (i.e., if $u$ is reachable from $v$ by a path in the realized network $G_T$). 
We  define the \emph{realized subnetwork} $G_{v,T}$  of an agent $v$ to be the induced subnetwork of the realized network $G_T$ composed of $v$ and all such $u$ reachable from $v$ in $G_T$.  
%, i.e., $V_{v,T}=\{u\in V\,:\, d_{G_T}(v,u)<\infty\}$. 
Whether $v$ learns the state is determined solely by his realized subnetwork: conditional on $G_{v,T}$, the action of $v$ is independent of the signals and arrival times of all the agents outside $G_{v,T}$.

Recall that $B_{r}(v)$ denotes the $r$-neighborhood of $v$ in the a-priori network $G$ and $\deg(v)$ denotes the degree of $v$ in $G$. The following lemma shows that the realized subnetwork of $v$ is contained in the neighborhood of $v$ with high probability provided that the radius of the neighborhood is big enough compared to the maximal degree.
\begin{lemma}[\new{Local nature of decisions}]\label{lm_localization_body}
	The probability  that the realized subnetwork of $v$ is contained in the $(r-1)$-neighborhood of $v$ enjoys the following lower bound:
	\begin{equation}\label{eq_localization_body}
		\P\Big(G_{v, T}\subset B_{r-1}(v)\Big)\geq 1-2\cdot\left(\frac{e\cdot \max_{u\in B_r(v)} \deg(u)}{r}\right)^{r}.
		\end{equation}
\end{lemma}
%Lemma~\ref{lm_localization_body} suggests that only a neighborhood of $v$ of radius $r$ having the order of the maximal degree increased by the factor of $e$ (the base of the natural logarithm) matters for $v$'s learning. In the next subsection, we find a necessary condition for $v$ to have high learning quality in terms of the topology of this neighborhood.

An extended version of the lemma is  proved in Appendix~\ref{app_localization}. Here we present a sketch.
\begin{proof}[Proof sketch for Lemma~\ref{lm_localization_body}]
The realized subnetwork of $v$ belongs to the $(r-1)$-neighborhood if and only if no path $(v=u_0,u_1,u_2,\ldots,u_k=u)$ connecting $v$ and the sphere  $B_r(v)\setminus B_{r-1}(v)$ in the a-priori network $G$ is presented in the realized network $G_T$. Note that a particular path of this form exists in $G_T$ if and only if $u_k$ arrives first, then $u_{k-1}$, then $u_{k-2}$, and so on. Hence, each path of length $k$ in $G$ remains in the realized network with  probability $\frac{1}{(k+1)!}$. Each path has length $k\geq r$ and the total number of paths of length $k$ connecting $v$ and the sphere in $G$ is bounded by $\big(\max_{u\in B_r(v)} \deg(v)\big)^k$. The union bound accompanied by manipulations with factorials similar to the Stirling formula lead to the desired inequality~\eqref{eq_localization_body}.
\end{proof}

\subsection{Local learning requirement}\label{subsec_local_requirement}

The previous subsection suggests that whether or not an agent $v$ learns the state must be determined by the local structure of the a-priori network $G$ around him. Here we formulate a \emph{local learning requirement} on the topology of $v$'s neighborhood that ensures high learning quality for $v$ no matter how well other agents perform and what the global structure of the network is. In subsequent sections we demonstrate the usefulness of this condition by constructing networks with exceptional learning and robustness properties. 

The essence of the requirement is that $v$ must bridge many  social circles. To define the requirement formally, \new{we recall the following notation:}  %and has
%a high-degree friend deep inside each of them. 
%Any connected subset $C\subset V$ is called a \emph{social circle}. %By $G\setminus v$ we denote the network obtained from the a-priori network by deletion of agent $v$  together with all adjacent edges, and by $B_{r,G\setminus v}(u)$, the $r$-neighborhood of an agent $u$ in $G\setminus v$. For an agent $v$, his friend $u$, and a social circle $C\ni u $, we say that $u$ is \emph{a friend of $v$ in $C$ of depth $r$} if $B_{r,G\setminus v}(u)\subset C$, i.e., if agent $v$ was eliminated, the whole $r$-neighborhood of $u$ would belong to $C$.
%\begin{definition}[Local learning requirement]
%An agent $v$  satisfies the local learning requirement with parameters $(d,r,D)$ if there are $d$ pairwise-disjoint social circles $C_1,C_2,\ldots,C_d$ such that in each social circle, $v$ has a friend of depth $r$ and degree at least $d$, and degrees of all agents from these social circles are bounded by $D$.
%\end{definition}
$G\setminus v$ is the network obtained from the a-priori network $G$ by elimination of an agent $v$  together with all adjacent edges and $B_{r,G\setminus v}(u)$ is the $r$-neighborhood of an agent $u$ in $G\setminus v$.
\begin{definition}[Local learning requirement]\label{def_loc_requirement}
	An agent $v$  satisfies the local learning requirement with parameters $(d,r,D)$ if among his friends we can find $d$ such that each of them has degree at least $d$, their $r$-neighborhoods in $G\setminus v$ are disjoint, and the degrees in all these neighborhoods are upper-bounded by $D$.	
	%Formally, there exists a subset $F_v^d\subset F_v$ of size $d$ such that, $\deg(u)\geq d$ for all $u\in F_v^d$, and $B_{r,G\setminus v}(u)\cap B_{r,G\setminus v}(u')=\emptyset$ for distinct $u,u'\in F_v^d$, and $\deg(w)\geq D$ for any $w\in \cup_{u\in F_v^d} B_{r,G\setminus v}(u)$.
\end{definition}
Let  $u_1,\ldots u_d$ be the friends from the definition. Then their neighborhoods $B_{r,G\setminus v}(u_i)$ can be interpreted as disjoint social circles bridged by $v$. The intuition why the definition refers to the network $G\setminus v$ relies on the fact that the decision of $v$ depends only on those agents that arrive before him and, hence, do not observe $v$ as if he was absent from the network.

Figure~\ref{fig_triplet} illustrates the definition. The following theorem is our main technical result. % and also shows that the same agent $v$ may satisfy the local learning requirement with different triplets $(d,r,D)$. 
%{We capture a local structure of the network around an agent $v$ by a triplet $(d, r, D)$ such that
%	\begin{itemize}
%		\item agent $v$ has at least $d$ friends with the degree $d$ or higher;	
%		\item the $r$-neighborhood  of $v$ is a tree;
%		\item each agent from the $r$-neighborhood has at most $D$ friends in $G$.
%\end{itemize}}
%This definition is illustrated by Figure~\ref{fig_triplet}.
%The local quality of decisions turns out to be related to this local structure of the network.
\begin{figure}
	\begin{center}	
	\includegraphics{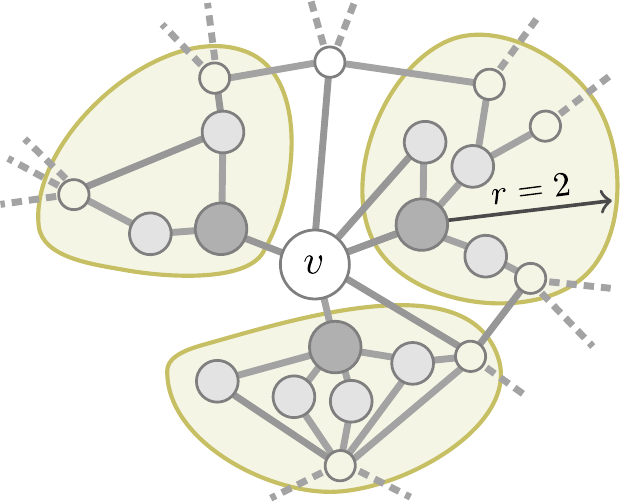}
	\end{center}
	\caption{Agent $v$ satisfies the local learning requirement with parameters $d=3$, $r=2$, and $D=7$. He has $3$ friends with degree at least $3$ (dark gray nodes), their $2$-neighborhoods in the graph $G\setminus v$ are disjoint (shaded areas), and the maximal degree in these neighborhoods is $7$ (the agent at the bottom has $5$ friends within the neighborhood and $2$ outside).
\label{fig_triplet}}
\end{figure}
\begin{theorem}\label{th_local_requirement}
	%Consider a network $G=(V,E)$ and an agent $v\in V$ with the triplet $(d, g_v, D)$.
	If an agent $v$ satisfies the local learning requirement with parameters $(d,r,D)$, then the learning quality of $v$ enjoys the following lower bound:
	\begin{equation}	\label{eq_local_quality_th_bound}
		l_\sigma(v)\geq 1-\delta(p, d, r, D),
		\end{equation}
		where
		\begin{equation}\label{eq_local_quality_th_delta_def}
		\delta(p, d, r, D)= \psi + \frac{18}{\sqrt{d-1}\,(2p-1-\psi)}, 	\ \ \ \ \ \ \ \psi=2d\cdot \left(\frac{e\cdot D}{r}\right)^{r}.
		\end{equation}
The bound holds for any state-symmetric equilibrium $\sigma$ provided that the probability of the correct signal satisfies $p\geq \frac{1+\psi}{2}$; i.e., the denominator in~\eqref{eq_local_quality_th_delta_def} is positive.\footnote{{The condition of state-symmetry can be dropped at the cost of getting $(4p-3-\psi)$ in the denominator of~\eqref{eq_local_quality_th_delta_def} instead of $(2p-1-\psi)$ and of imposing the stricter condition on $p$ to ensure the positivity of the new denominator (see the discussion on the technical assumptions in Section~\ref{sec_discuss}).}}
	%provided that $g_v> 5$ and $\varepsilon<\frac{1-\psi}{2}$.}
\end{theorem}
}
%\edw{Without the assumption of state-symmetric equilibrium, one can prove an analog of this proposition with $4\varepsilon$ instead of $2\varepsilon$ in denominator.}
\begin{corollary}\label{cor_delta_to_zero}
{The theorem shows that the agent learns the state whenever $\delta(p, d, r, D)$ is close to zero. We note that $\delta(p, d, r, D)$ goes to zero for fixed $p>\frac{1}{2}$ when all the elements of the triplet $(d, r, D)$  go to infinity such that $\lim\inf \frac{r}{D}> e$.} 	
	% has many high-degree friends and a tree-like neighborhood that is big enough compared to the maximal degree.
	%It is quite surprising that the radius of the same order as the maximal degree is enough.
\end{corollary}	
{The intuition behind Theorem~\ref{th_local_requirement} is simple. \new{Let $u_1,\ldots, u_d$ be the friends of $v$ from Definition~\ref{def_loc_requirement}. By the \new{local nature of decisions} (Lemma~\ref{lm_localization_body}), the realized subnetwork of each $u_i$, who arrives before $v$, is likely to be contained within his social circle $B_{r,G\setminus v}(u_i)$.} Hence, the realized subnetworks are disjoint, which ensures no herding among $(u_i)_{i=1,\ldots,d}$. As a result, $v$ learns the state by observing a large sample of independent sources of information.}
However, formalization of this intuition faces some obstacles: the independence is only conditional and the conditioning is on a family of events that do not belong to the information partition of any of the agents; hence, checking that the independent sources are \ed{informative} requires the approximation of the condition by elements of information partitions; we are able to carry out this approximation only for  high-degree friends of $v$ since their partitions are finer {(hence, the condition $\deg(u_i)\geq d$).}

{Here we present a sketch of the proof; all the details can be found in Appendix~\ref{appendix_local_requirement}.}
	%Now we discuss the consequences of the proposition, which are used in the proof of Theorem~\ref{th_egalitarian} and in Section~\ref{sect_robust}.
	%Instead of the radius $r$ of the tree-like neighborhood, which is a local parameter,  it is convenient to use a global one: the girth.
\begin{proof}[Sketch of the proof of Theorem~\ref{th_local_requirement}]
 \new{Denote by $F_{v,T}^d$ the subset of those friends $u_1,\ldots, u_d$  from Definition~\ref{def_loc_requirement} who arrive earlier than $v$.}
Consider the following deviation $a_v'$ from an equilibrium strategy $a_v$ of  $v$. He repeats the action played by the majority of $F_{v,T}^d$.
	In equilibrium, the agent cannot benefit from this deviation and, hence,
		$l_\sigma(v)\geq \P(a_v'=\theta)$.
	
	{The probability of a mistake for $a_v'$ can be bounded using the Hoeffding inequality.\footnote{\label{footnote_Hoeffding}The Hoeffding inequality~\citep{hoeffding1994probability} states that $\P\left(\frac{1}{N}\sum_{n=1}^N \xi_i\leq \E\left(\frac{1}{N}\sum_{n=1}^N \xi_i\right)-x \right)\leq \exp\left(- 2 x^2 N\right)$ for independent random variables $0\leq \xi_i\leq 1$ and any $x\geq 0$.}  This requires two ingredients: independence of the actions $a_u$ for $u\in F_{v,T}^d$ and a lower bound on the probability of a mistake for $a_u$.} %Both are the subjects of Lemma~\ref{lm_technical} from the appendix.}

	\ed{Unfortunately, the requirement of independent actions is not satisfied.} 
	However, \ed{the actions} can be made independent by conditioning on a certain collection of events. This should still be enough to drive the result \ed{provided that} the probability \ed{of this collection  is close to one}. 	
	Let $W_v$ be the event that  the realized subnetwork of each $u_i$ who arrived before $v$ is contained in his social  circle $B_{r,G\setminus v}(u_i)$. 
	Since these social circles are disjoint, the realized subnetworks do not intersect and thus each $u_i$ aggregates the information from disjoint families of  sources conditional on $W_v$.	\ed{The} \new{local nature of decisions} (Lemma~\ref{lm_localization_body})  implies that the probability of  $W_v$ is close to one.
	
	Naively, one could argue that conditional on $W_v$ the actions taken by the friends of $v$ are independent. \ed{However,} this argument is wrong.  There are other sources of dependence. For example, the fact that we are interested in agents who arrive \emph{before $v$} creates dependence between their actions even if their realized subnetworks are disjoint. To see this consider the case where $v$ arrives early. This implies that all the friends he observes are also early arrivals and, hence, are likely to have smaller realized subnetworks compared to the case where $v$ arrives late; in particular, the earlier $v$ arrives, the more mistakes the observed friends make, thus creating  dependence between their decisions.
	
	To eliminate all the sources of dependence, we end up conditioning on the arrival time of $v$, the set $F_{v,T}^d$, the realized state, and the event $W_v$.
	
	{In order to apply the Hoeffding inequality, it remains to show that the conditional probability of a mistake for $a_u$ is bounded away from $1/2$ for $u\in  F_{v,T}^d$, i.e., that the independent sources observed by $v$ are informative. We use the following idea. For any event $A$ from the information partition of $u$, we have $\P(a_u\ne \theta\mid A)\leq 1-p$ because otherwise $u$ is better off  following his signal whenever $A$ occurs; additional conditioning on the realized state does not change the bound because we assume a state-symmetric equilibrium.	
		Unfortunately, the family of events on which we condition
 does not belong to $u$'s information partition. We overcome this difficulty by approximating the condition to elements of the information partition and showing that  the conditional probability of a mistake is at most $\frac{1}{1-\psi}\left(1-p + \frac{3}{t_v\cdot \sqrt{\deg(u)-1}}\right).$
		The bound gets worse for low-degree agents since their information partition is not fine enough for a good approximation. This is why the deviation $a_v'$ of agent $v$ takes into account high-degree friends only.}
	
	{After these preparations, Theorem~\ref{th_local_requirement} becomes a corollary of the Hoeffding inequality.}	
\end{proof}	

\section{Implications: Symmetric networks that support learning}\label{sec_egalitarian}

{Here we demonstrate that there are a-priori networks where  \emph{each agent} satisfies the local learning requirement from the previous section and, hence, all agents achieve high learning quality as does the network itself. Moreover, there are such networks 
 with the additional property that any two agents play the same role.} %which is captured by the following definition.
\begin{definition}[Symmetric networks]
A network $G=(V,E)$ is symmetric if for any pair $v,v'\in V$  {there exists an automorphism $f$ of $G$  such that\footnote{In mathematical literature, such graphs are usually called \emph{transitive} since the group of automorphisms acts transitively on them; i.e., for any pair of vertices there is an automorphism mapping one to the other.} $f(v)=v'$.}
\end{definition}
{Symmetric networks represent totally egalitarian societies.  \ed{The celebrity graphs from  Example~\ref{ex_celebrity} may lead to a conjecture that egalitarian societies cannot aggregate information since one needs a designated minority of agents (like celebrities) in the a-priori network for learning to propagate.
	The main result of this section states that this intuition is false and symmetric networks  can support learning.} In Section~\ref{sec_robust},  
we will strengthen this result by showing robustness of learning.
\begin{theorem}\label{th_egalitarian}
For any $p_0>\frac{1}{2}$ and any $\delta>0$ there exists a symmetric network	$G=(V,E)$ such that the learning quality $l_\sigma(v)\geq 1-\delta$ for each agent $v\in V$, any probability of the correct signal $p>p_0$, and any {state-symmetric} equilibrium\footnote{The theorem extends to non-state-symmetric equilibria at the cost of assuming that signals are informative enough, namely, $p_0>\frac{3}{4}$. This and other extensions are discussed in Section~\ref{subsec_getting_rid}.} $\sigma$.
%Moreover, for any $0<p<1$ and any subset of agents $V'\subset V$ with $p\cdot|V|\leq |V'| <  p\cdot|V|+1$, in the subnetwork $G^{V'}$ spanned by $V'$ the average fraction of agents with correct actions is at least~$1-\frac{2}{p^3}\delta$.
\end{theorem}
In order to prove Theorem~\ref{th_egalitarian}, we use classic results from the theory of expanders to construct networks, where the local  learning requirement is satisfied for each agent. 

We will need the following lemma, which simplifies checking the local learning requirement. Recall that the \emph{girth} $g$ of a network $G$ is the length of the shortest cycle; the girth of a tree is infinite, $g=+\infty$.
%By $\lfloor x\rfloor$ we denote the floor of a real number~$x$. 
\begin{lemma}\label{lm_girth_and_llr}
Consider a network $G$ of girth $g$ and the maximal degree $D$. In such a network, any agent $v$ satisfies the local learning requirement with parameters $\Big(d,\Big\lfloor\frac{g-3}{2}\Big\rfloor,D\Big)$, where $d$ is a number such that $v$ has at least $d$ friends of degree $d$ or more.
\end{lemma}	
\new{\begin{proof}
Let $u_1,u_2,\ldots, u_d$ be $v$'s friends with degrees  $\deg(u_i)\geq d$.
  Recall that $B_{r,G\setminus v}(u)$ denotes the $r$-neighborhood of $u$ in the network, where the agent $v$ was eliminated. We aim to select $r$ such that $B_{r,G\setminus v}(u_i)$ and $B_{r,G\setminus v}(u_j)$ are disjoint for $i\ne j$. If $B_{r,G\setminus v}(u_i)$ and $B_{r,G\setminus v}(u_j)$ intersect, this creates a cycle of length $2r+2$ (or shorter): start from $v$, then go to $u_i$, then to an agent in the intersection by the shortest path, then to $u_j$ again by the shortest path, and back to $v$. Hence, the intersection is possible only if $2r+2\geq g$. Thus choosing  $r$ such that $2r+2<g$, we ensure that the $r$-neighborhoods are disjoint;  $r=\Big\lfloor\frac{g-3}{2}\Big\rfloor$ is the maximal such $r$. We conclude that $v$ satisfies the local learning requirement with parameters~$\Big(d,\Big\lfloor\frac{g-3}{2}\Big\rfloor,D\Big)$.
\end{proof}}}
Recall that a network is called  \emph{$D$-regular} if all vertices have degree $D$.
Thanks to Lemma~\ref{lm_girth_and_llr} and Theorem~\ref{th_local_requirement}, proving Theorem~\ref{th_egalitarian} reduces to the question of the existence of symmetric $D$-regular networks with arbitrary high degree and girth. The existence of such networks is demonstrated in the theory of expanders; see Appendix~\ref{sect_expanders}. The following is the corollary of a theorem by \citet{lubotzky1988ramanujan} (Theorem~\ref{th_Ramanujan} in Appendix~\ref{sect_expanders}) that describes a family of so-called Ramanujan expanders.
\begin{corollary}[of Theorem~\ref{th_Ramanujan} by \citet{lubotzky1988ramanujan}]\label{cor_existence_high_girth}
	There is a sequence $D_k\to \infty$ such that for any $g_0$ and $k$, there exists a symmetric  $D_k$-regular network $G=(V,E)$  such that the girth $g\geq g_0$ and $|\lambda_2|\leq 2\sqrt{D_k-1}$, where $\lambda_2$ is the second-largest eigenvalue of the adjacency matrix.\footnote{We do not use the bound on $\lambda_2$ in Section~\ref{sec_egalitarian}; in particular, any $D$-regular network with large $D$ and  girth would suffice for the proof of Theorem~\ref{th_egalitarian}. The bound on $|\lambda_2|$ will be critical for the discussion on robust learning in Section~\ref{sec_robust}.}
\end{corollary}	
After all these preliminaries, proving Theorem~\ref{th_egalitarian} becomes easy.
{\begin{proof}[Proof of Theorem~\ref{th_egalitarian}]
  For a $D$-regular network of girth~$g$, Theorem~\ref{th_local_requirement} combined with  Lemma~\ref{lm_girth_and_llr} implies a lower bound on the learning quality $l_\sigma(v)\geq 1-\delta\left(p,D,\Big\lfloor\frac{g-3}{2}\Big\rfloor,D\right)$ for any agent $v$ and any state-symmetric equilibrium $\sigma$. 

Corollary~\ref{cor_existence_high_girth} allows us to pick a symmetric $D$-regular network $G$ with arbitrary high degree $D$ and arbitrary high girth/degree ratio $\frac{g}{D}$, while Corollary~\ref{cor_delta_to_zero} implies that $\delta\left(p,D,\Big\lfloor\frac{g-3}{2}\Big\rfloor,D\right)$ tends to zero if both the girth $g$ and the degree $D$ tend to infinity and the girth goes to infinity faster: $\frac{g}{D}\to \infty$.
Hence, for any given $\delta$ and $p_0$,  we can choose $G$ to be a symmetric $D$-regular network with $D$ and $g$ such that $\delta\left(p_0, D, \Big\lfloor\frac{g-3}{2}\Big\rfloor, D\right)\leq\delta$. Since, $\delta$ is decreasing in $p$, we obtain
$ \delta\left(p,D,\Big\lfloor\frac{g-3}{2}\Big\rfloor,D\right)\leq \delta\left(p_0,D,\Big\lfloor\frac{g-3}{2}\Big\rfloor,D\right)\leq \delta$ for any $p\geq p_0$.  Thus the network $G$ satisfies the statement of Theorem~\ref{th_egalitarian}.
\end{proof}	
}

{
\section{Implications: Robust learning}\label{sec_robust}
Here we discuss the learning quality of a network when some of the agents 
are uninterested or unavailable and, hence, do not participate in the learning process.
The celebrity graphs  (Example~\ref{ex_celebrity}) demonstrate that  high learning quality can be very fragile: if celebrities  (a negligible minority of the population) do not show up, this leaves the network totally disconnected, and information aggregation breaks down. We show that there are networks that \emph{support robust learning}, i.e., are free of such a flaw.

In Section~\ref{sec_egalitarian}, we saw that symmetric sparse networks that have high degrees but not short cycles support learning. In such networks, all agents play the same role, which makes it  natural to expect that no \emph{small group} of agents is critical in these networks; in particular, adversarial elimination of a small group cannot dramatically spoil the learning outcome.
Here we obtain the surprising, much stronger result demonstrating that there are symmetric networks where learning is robust to the adversarial elimination of any group, even a large one.
}
\begin{theorem}\label{th_robust}
	For any $p_0>\frac{1}{2}$ and any $\delta>0$ there exists a symmetric network	$G=(V,E)$ such that for any $\alpha \in (0,1]$ and any subset  $V'\subset V$ with $\big\lceil \alpha\cdot|V|\big\rceil$ agents, the learning quality in the induced subnetwork $G^{V'}$  is at least $1-\frac{\delta}{\alpha^3}$ for any state-symmetric equilibrium and any probability of the wrong signal\footnote{{Similarly to Sections~\ref{sec_local_structure} and~\ref{sec_egalitarian}, the result extends to non-state-symmetric equilibria under the requirement $p_0>\frac{3}{4}$.}} $p>p_0$.
\end{theorem}

Robustness of learning turns out to be related to spectral properties of the network. This is captured by Lemma~\ref{lm_robust} below, and Theorem~\ref{th_robust} easily follows from this lemma combined with known results on expander graphs. Consider a $D$-regular network and denote by $|\lambda_1|\geq |\lambda_2|\geq...\geq |\lambda_{|V|}|$ the eigenvalues of its adjacency matrix ordered by their absolute values. Expanders are networks with small $|\lambda_2|$ relative to $|\lambda_1|$; see Appendix~\ref{sect_expanders}.

{The next lemma bounds by how much the learning quality can decrease if, instead of the original $D$-regular network, we consider its arbitrary subnetwork with $\lceil \alpha\cdot |V|\rceil$  agents.}
\begin{lemma}\label{lm_robust}
	Let $G=(V,E)$  be a $D$-regular network  of girth $g$ with the second-largest eigenvalue $\lambda_2$. Then
	%\point for any agent $v\in V$ the chance of making the correct action satisfies
	%\begin{equation}\label{eq_learning_quality_expanders}
	%\P(a_v=\theta)\geq 1-\delta, \ \  \mbox{  where  } \ \ \ \ \delta=\frac{7\cdot \left(\frac{2 D\cdot e}{g-1} \right)^{\frac{g-1}{2}}}{1-4\varepsilon}+\frac{18}{(1-4\varepsilon)^2\cdot D}.
	%\end{equation}
	%\point
	for any $\alpha\in \left(0, 1\right]$ and any subset $V'\subset V$ of size $|V'|=\big\lceil \alpha\cdot|V|\big\rceil$,  the learning quality in the induced subnetwork $G^{V'}$ satisfies
	\begin{equation}\label{eq_average_quality_subexpander}
	L_{\sigma'}\left(G^{V'}\right)\geq 1-\left(\frac{2}{\sqrt{\alpha} }+\frac{(1-\alpha)|\lambda_2|^2}{ \alpha^3 \cdot D^{\frac{3}{2}}}\right)\cdot {\delta\left(p,D,\Big\lfloor\frac{g-3}{2}\Big\rfloor,D\right)}
	\end{equation}
	 for any state-symmetric equilibrium $\sigma'$. Here $\delta(p, d,r,D)$ is given by formula~\eqref{eq_local_quality_th_delta_def}.
%	 {$\delta=\psi+\frac{1}{D}\cdot \frac{1}{1-\exp\left(-\frac{1}{2}\frac{(1-\psi-2\varepsilon)^2}{(1-\psi)^2}\right)}$ and
%	 $\psi=D\cdot \left(2e\cdot\frac{D}{g-5}\right)^{\frac{g-5}{2}}$.}
\end{lemma}	
Lemma~\ref{lm_robust} is proved in Appendix~\ref{appendix_missed_proofs_robust}; here we present the main idea and then prove Theorem~\ref{th_robust}.
 \begin{proof}[Sketch of the proof of Lemma~\ref{lm_robust}.]
{The key tool is the mixing lemma from the theory of expanders (Lemma~\ref{lm_mixing}). This lemma is applicable to any $D$-regular network and provides a bound in terms of $|\lambda_2|$ on how much the number of edges $|E(V^1,V^2)|$ between any two disjoint subsets of vertices $V_1,\ V_2$ deviates from the expected number of edges in the Erdos\textendash Renyi model:
\begin{equation}\label{eq_mixing}
\left|  \big|E(V^1,V^2)\big| -\frac{D}{|V|}\cdot|V^1|\cdot|V^2| \right|\leq |\lambda_2|\sqrt{|V^1||V^2|}.
\end{equation}
This lemma allows us to bound the fraction of agents with low degree in $G^{V'}$. Indeed, for $\gamma<\alpha$, let $V^1$ be the set of agents with degrees less than $\gamma\cdot  D$ in $G^{V'}$ and $V^2$ be the set of eliminated agents $V\setminus V'$. Since each agent has degree $D$ in the original network $G$, there are at least $\left(1-\gamma\right)D|V^1|$ edges between $V^1$ and $V^2$. The Erdos\textendash Renyi model prescribes $(1-\alpha)D|V^1|$ edges in expectation. This discrepancy is compatible with~\eqref{eq_mixing} only for relatively small subsets $V^1$. We obtain that the fraction of agents in $G^{V'}$ with degree $\gamma \cdot D$ or higher is at least $1-\theta$, where $\theta=\theta(\alpha,D,\gamma,\lambda_2)$ is small.
}

By Lemma~\ref{lm_friends_with_high_degree} from the appendix, the lower  bound on the fraction of high-degree agents implies a bound on the number of agents having a large number of high-degree friends: the fraction of agents  having at least $\frac{\gamma}{2}D$ friends with degrees $\frac{\gamma}{2}D$ and higher is above $1-2\theta$. {By Lemma~\ref{lm_girth_and_llr}, each such agent satisfies the local learning requirement with parameters $\Big(\frac{\gamma}{2}D, \Big\lfloor\frac{g-3}{2}\Big\rfloor ,\frac{\gamma}{2}D\Big)$. Application of Theorem~\ref{th_local_requirement} to this set of agents completes the proof.}
\end{proof}	
%Theorem~\ref{th_robust} easily follows from Proposition~\ref{prop_robust} and known properties of the Ramanujan expanders.
\begin{proof}[Proof of Theorem~\ref{th_robust}]
	{As in the proof of Theorem~\ref{th_egalitarian}, we can pick
	the Ramanujan expander {such that $\delta\left(p_0,D,\Big\lfloor\frac{g-3}{2}\Big\rfloor,D\right)$  is at most $\frac{\delta}{6}$.} The second eigenvalue of the Ramanujan expander satisfies $|\lambda_2| \leq 2\sqrt{D-1}\leq 2\sqrt{D}$ (Corollary~\ref{cor_existence_high_girth}) and, hence, by Lemma~\ref{lm_robust},
	any subnetwork $G^{V'}$ with $|V'| =\big\lceil  \alpha\cdot|V| \big\rceil$ has  learning quality  at least $$1-\left(\frac{2}{\alpha}+\frac{4(1-\alpha)}{\alpha^3\sqrt{D}}\right)\cdot\delta\left(p,D,\Big\lfloor\frac{g-3}{2}\Big\rfloor,D\right).$$
	Since $\frac{2}{\alpha}+\frac{4(1-\alpha)}{\alpha^3\sqrt{D}}\leq \frac{2}{\alpha^3}+\frac{4}{\alpha^3}\leq \frac{6}{\alpha^3}$ and $\delta\left(p,D,\Big\lfloor\frac{g-3}{2}\Big\rfloor,D\right)\leq \delta\left(p_0,D,\Big\lfloor\frac{g-3}{2}\Big\rfloor,D\right)$, the learning quality is bounded from below by $1-\frac{\delta}{\alpha^3}$.}
\end{proof}
{
\begin{remark}[Randomized robustness]
	Instead of robustness with respect to adversarial elimination, one can consider a weaker notion: a subset of agents is deleted at random, and the resulting network has to support learning with high probability with respect to the choice of this subset.
	
	In Appendix~\ref{app_randomized_robustness}, we show that, surprisingly, any network that supports learning has the property of randomized robustness. For example, in the celebrity graphs, if we eliminate a large set of agents randomly, say $50\%$, around $50\%$ of the celebrities will remain in the network, thus ensuring information aggregation (with slightly lower learning quality).

	More formally, for a network $G=(V,E)$ consider an induced subnetwork $G^{V'}$, where $V'$ is a subset of $\lceil\alpha\cdot |V|\rceil$ agents taken uniformly at random. Under an additional technical assumption, we demonstrate  that $L(G^{V'})\geq 1- \sqrt{\frac{1-L(G)}{\alpha}}$ with probability at least $1- \sqrt{\frac{1-L(G)}{\alpha}}$ with respect to the choice of $V'$ (Theorem~\ref{th_random_elimination}); here  $L(G)=\max_\sigma L_\sigma(G)$ denotes the learning quality for the best equilibrium.
	
	Randomized robustness holds because, in a sense, it is hardwired in the definition of learning quality with a random arrival order. Indeed, if a network aggregates information well for most  arrival orders, then the subnetwork of the first $\lceil\alpha\cdot |V|\rceil$ arrivals must also aggregate information well for most such networks. The proof of  Theorem~\ref{th_random_elimination} is based on this coupling between the choice of the subset $V'$ and the learning process in the original network.
\end{remark}	
}

%\section{Implications: fake news}\label{sec_fake_news}
%\fed{Kill?!?}

\section{Discussion}\label{sec_discuss}

{In Section~\ref{subsect_interp}, we offer two alternative ways to perceive the results on learning in symmetric networks (Section~\ref{sec_egalitarian}) and those on robust learning (Section~\ref{sec_robust}). First, we present the network-designer perspective and then explain the connection to the classic sociological theory of the two-step information flow.
Section~\ref{subsec_getting_rid} discusses our technical assumptions and how to relax them.}

\subsection{Alternative interpretations of results}\label{subsect_interp}
\paragraph{Network-designer perspective} 

%For a network-designer, our results offer a recipe for building networks that aggregate information even if a subset of agents do not participate in the learning process. This subset of absent agents can be large and picked in an adversarial way.

Consider a  designer of a social network who wants to ensure social learning. \citet{sgroi2002} constructs a network that supports learning for a particular arrival order, and \citet{bahar2019celeb} propose the celebrity graphs as  networks that support learning and are robust to the arrival order of agents. However, the celebrity graphs only support learning if {\em all} agents actually participate and make a decision.  In other words, if a minority of agents do not care about the decision at hand and are consequently inactive, then learning may fail. By contrast, our expander-based networks from Theorem~\ref{th_robust} give a  social structure that supports learning and is robust both to the arrival order and to the actual subset of active agents.

{In practice, social networks are not designed from scratch and the social connections can be considered as given. However, social networks on the Internet are usually combined with a recommendation system that decides which news to show to a particular user. In so doing, it can intervene in the structure of social connections by concealing some friends' news and possibly showing some news of non-friends. Our local learning requirement (Theorem~\ref{th_local_requirement}) combined with Lemma~\ref{lm_girth_and_llr} suggests a possible recipe for improving overall learning quality: eliminating short cycles while keeping high the number of information sources an agent is subject to. This proposal, however, requires additional empirical evaluation.}
\fed{We can mention fake news  here. Should we?}
 %Can such systems be designed in a way that facilitates social learning, e.g., make it robust to no-show of some agents? This question is close to the modern line of research on recommendation systems recognizing externalities they create: it can be reasonable to serve some agents with individually sub-optimal content in order to achieve greater societal benefit (e.g., \citep{basat2017game, bahar2019fiduciary}).

\paragraph{The paradigm of the two-step information flow}

The classic sociological theory of the ``two-step information flow'' by \citet{katz1955personal} conveys the idea that learning is always facilitated by a small group of \ed{opinion leaders, i.e., influential agents  predetermined by the network structure.} The original version of the theory, formulated during the golden era of TV networks, argues that although such networks were responsible for sparking new ideas and introducing new products, this content was consumed by most people in an indirect manner. In other words, people's actions, whether in the form of adopting a new product or voting for a certain political candidate, were not so much a result of what they heard from the TV networks but rather what they heard from influential agents who, for their part, had consumed such content directly.

{Let us call a group of agents \emph{influential} if the network supports learning whenever this group is present and fails to aggregate the information when it is not. An example of an influential minority is the group of celebrities in the celebrity graphs of Example~\ref{ex_celebrity}.

The theory of the two-step information flow suggests that networks  supporting learning contain an influential minority, which plays the role of intermediary in the information spread over the network. The results of Sections~\ref{sec_egalitarian} and~\ref{sec_robust} challenge this thesis (in our stylized theoretical model): high learning quality can be achieved in a-priori symmetric networks and, hence, no group of agents is predetermined by the network structure; moreover, there are symmetric networks where no minority (or even majority) of agents is influential.}
\fed{Do we want to add anything?}

\subsection{Getting rid of technical assumptions}\label{subsec_getting_rid}
Some of our modeling assumptions were made to simplify the exposition {and can be easily relaxed.} %Relaxing the following assumptions will not alter our results and the conclusions.
\paragraph{General distributions of arrival times} We assumed that agents' arrival times are i.i.d. and, in particular, uniformly distributed on the unit interval. However, any non-atomic distribution leads to an equivalent model (equivalence is obtained by a monotone reparameterization), as long as the i.i.d. assumption is maintained. An important robustness result would be to extend our conclusions to some approximate notion of i.i.d. as some local dependence among agents is a realistic assumption.

%In reality, independence is satisfied only approximately, e.g., both the timing of $v$'s decisions and the local structure of the network around $v$ may depend on  $v$'s socio-economic attributes.
	%We ignore these and other realistic complications (e.g.,  network externalities inevitable in real-life decisions) in order to keep the model tractable and focus on information-aggregation exclusively.}
%In particular, all of the results extend to signals with possibly unbounded precision.	
%However, it is critical for our results that all the signals have bounded precision (the so-called bounded signals). This assumption ensures that the only way an agent can reduce the probability of mistake is by aggregating the information contained in his friends' actions.}	
\paragraph{Non-binary signals}	
	Our results hold for any non-binary signaling device as long as signals are informative and symmetry is maintained. By informativeness we mean that for a positive probability set of signals $s$, the posteriors $\P(\theta=1\mid s)$ belong to the union of intervals $[0,1-p]\cup [p,1]$ (the probability of this set of signals will enter into the bound on the learning quality). By ``symmetry'' we mean that the distribution posteriors $\P(\theta=1\mid s)$  is symmetric around $\frac{1}{2}$. Note that signals of unbounded precision are also allowed. 
 \paragraph{Heterogeneous agents}
 {Agents can be heterogeneous, i.e., the signaling device can be agent-specific, and so each agent $v$ has his own signal precision $p_v$. In this case, all the results hold with $p=\min_{v\in V}p_v$.}
 \paragraph{Non-symmetric equilibria}
 {The symmetry assumption on equilibria and signals can  be relaxed. However, for asymmetric equilibria, Theorems~\ref{th_local_requirement}, \ref{th_egalitarian}, and~\ref{th_robust} require the probability $p$ of the correct signal to be above $\frac{3}{4}$ (instead of $\frac{1}{2}$). The reason for this is  inequality~\eqref{eq_conditional_mistake_bound}, where for asymmetric equilibria we get $2(1-p)$ in the parentheses.}
\paragraph{General states} Extension to non-equiprobable states is straightforward. Note, however, that this breaks the state symmetry of equilibria and, hence, we get $2(1-p)$ in inequality~\eqref{eq_conditional_mistake_bound} (see the comment above). Extension to a non-binary state does not lead to any additional technical difficulties.
\paragraph{More informed agents} The results of Sections~\ref{sec_local_structure}, \ref{sec_egalitarian}, and~\ref{sec_robust} hold  if
in addition to observing the actions of his friends, each agent $v$
gets some information about his realized subnetwork $G_{v,T}$, e.g., a possibly noisy signal about the set of agents, their arrival times, actions, and even their private signals.
 \paragraph{Less informed agents}  {All our results  
 	can be easily adapted to the case, where observation of friends' actions is noisy. Namely, each agent $v$, instead of observing the action $a_u$ of his friend $u\in F_{v,T}$, observes either $a_u$ with probability $1-\varepsilon$ or the flipped action $1-a_u$ with probability $\varepsilon$. One can consider two variants of this model: the action of $u$ is flipped for all his neighbors at the same time but independently across $u\in V$, or the action is flipped independently for each pair $(v,u)$.
 	
 	Both variants require the same straightforward modifications in Sections~\ref{sec_local_structure}, \ref{sec_egalitarian}, and~\ref{sec_robust}. They originate 
 	from a minor adjustment in the proof of Theorem~\ref{th_local_requirement}: when applying the Hoeffding inequality in formula~\eqref{eq_Hoeffding_proof_th_local}, instead of the upper bound on $\P(a_u\ne \theta)$, we will need an upper bound on the probability that the action \emph{observed by $v$} does not match the state. The latter probability is equal to $(1-\varepsilon)\P(a_u\ne \theta)+\varepsilon(1-\P(a_u\ne \theta))$ and does not exceed  $(1-\varepsilon)\P(a_u\ne \theta)+\varepsilon$, where $\P(a_u\ne \theta)$ can be bounded by Lemma~\ref{lm_upper_bound_mistake} as before.	
}

\addcontentsline{toc}{section}{\protect\numberline{}References}%

\bibliographystyle{plainnat}

\bibliography{bibliography_robust}

\appendix

\section{Expanders}\label{sect_expanders}
%\subsection{Graphs that locally look like a tree: Ramanujan expanders}

%If a large enough neighborhood of a vertex $v$ is a tree, then, by Lemma~\ref{lm_local}, no friends observed by $v$ have common predecessor with high probability. Thus we can use Lemma~\ref{lm_no_predecessors} for showing that $v$ is likely to makes the correct action.
%In order to use this strategy for proving Theorem~\ref{th_egalitarian}, we need a symmetric graph $G$ with arbitrary large degree $D$ such that the neighborhood of each agent is a tree, i.e., the graph has no short cycles.
%The girth $g$ of a graph $G=(V,E)$ is the length of a shortest cycle. Existence of graphs with high girth and high minimal degree is itself a non-trivial question, which is resolved in the theory of expanders.

%Expanders are sparse graphs

There are several equivalent definitions of expanders; see~\citet{hoory2006expander}. The ``spectral'' definition is the most convenient for our needs. Consider a $D$-regular ($\deg(v)=D$ for all vertices) graph $G=(V,E)$ and denote by $(\lambda_k)_{k=1,\dots,|V|}$  eigenvalues of its adjacency matrix ordered by absolute values $|\lambda_1|\geq|\lambda_2|\geq|\lambda_3|....\geq |\lambda_{|V|}|$. The eigenvalue $\lambda_1=D$ corresponds to the eigenvector representing the uniform distribution over vertices: the smaller  the second eigenvalue  $|\lambda_2|$ is, the faster the distribution of a random walk started at some vertex converges to the uniform distribution; see Proposition~1.6. in~\citet{lubotzky2012expander}. A graph $G$ is an \emph{expander} if $|\lambda_2|$ is small relative to $|\lambda_1|$; i.e., a random walk on $G$ forgets the starting point fast.
\begin{theorem}[\citet{lubotzky1988ramanujan}]\label{th_Ramanujan}
	For any $N$ and $D$ such that $D-1$ is a prime number and $D-2$ is divisible by $4$, there exists a $D$-regular symmetric  graph $G=(V,E)$   with at least $N$ vertices, $|\lambda_2|\leq 2\sqrt{D-1}$, and girth $g\geq \frac{2}{3}\log_{D-1} |V|$.
\end{theorem}
Graphs with $|\lambda_2|\leq 2\sqrt{D-1}$ are called Ramanujan graphs; by the Alon\textendash Boppana theorem (see Theorem~5.3 in~\citet{hoory2006expander}), this value of $\lambda_2$ is essentially the best possible. It is quite intuitive that the best expanders cannot have short cycles, which lead to recurrences in the random walk and, hence, slow down the expansion of the random walk over the graph.

The proof of Theorem~\ref{th_Ramanujan} is quite technical and relies on group theory. \citet{lubotzky1988ramanujan} and  their successors (e.g., \citet{morgenstern1994existence}, who relaxed the condition  of divisibility by $4$, and~\citet{dahan2014regular}, who extended the result to arbitrary $D\geq 11$) construct $G$ as a Cayley graph of a certain group.\footnote{Given a group $\mathcal{G}$ with a group operation $\circ$ and a subset $S\subset \mathcal{G}$, its Cayley graph is defined in the following way: the set of vertices $V$ coincides with $\mathcal{G}$  and  $vu\in E$ if $u=v\circ s$ for some $s\in S$.} While the symmetry of the constructed graph is not mentioned explicitly in these papers, it comes for  free because any Cayley graph is symmetric; see Claim~11.4 in \citep{hoory2006expander}.\footnote{Indeed, for any $x\in \mathcal{G}$  the map \ed{$f(v):\, v\to x\circ v$ is an automorphism of the Cayley graph. For any pair of vertices $v,v'\in V$, we achieve $f(v)=v'$ by choosing $x=v'\circ v^{-1}$.}  }

In addition to large girth, %in Section~\ref{sect_robust}
we need another property of expanders demonstrating their similarity to random graphs in the Erdos\textendash Renyi model with the probability $p$ of an edge between two given vertices  equal to $\frac{D}{|V|}$. For a pair of disjoint subsets $V^1,V^2\subset V$ denote by $E(V^1,V^2)$ the set of edges with one endpoint in $V^1$ and another in $V^2$. In the Erdos\textendash Renyi model the expected  number of such edges is equal to $p\cdot|V^1||V^2|$. The following result, known as the mixing lemma, shows that $\big|E(V^1,V^2)\big|$ for an expander is close to this number.
\begin{lemma}[Mixing lemma, \citet{alon1988explicit}] \label{lm_mixing}
	For any $D$-regular graph $G=(V,E) $ and any two disjoint subsets $V^1,V^2\subset V$ the following inequality holds:
	$$\left|  \big|E(V^1,V^2)\big| -\frac{D}{|V|}\cdot|V^1|\cdot|V^2| \right|\leq |\lambda_2|\sqrt{|V^1||V^2|}.$$
\end{lemma}

\section{\new{Local nature of decisions}}\label{app_localization}
\ed{Here we prove Lemma~\ref{lm_localization_body}, which ensures  that the realized subnetwork $G_{v,T}$ of $v$ is likely to be contained in the neighborhood of $v$ whenever the radius of this neighborhood is large enough compared to the maximal degree. In fact, we prove a slightly more general statement with conditioning on $v$'s arrival time; this is a technical nuance that we will need in the proof of Theorem~\ref{th_local_requirement}.
%\ed{Recall that $B_r(v)$ denotes the $r$-neighborhood of $v$ in the a-priori network $G$.}
\begin{lemma}\label{lm_localization_appendix}
	The probability  that the realized subnetwork of $v$ is contained in the $(r-1)$-neighborhood of $v$,  conditional on $v$'s arrival time $t_v$, enjoys the following lower bound:
	\begin{equation}\label{eq_localization_appendix}
	\P\Big(G_{v,T}\subset B_{r-1}(v)\mid  t_v \Big)\geq 1-2\cdot\left(\frac{e\cdot D}{r}\right)^{r},
	\end{equation}
	where   $D=\max_{u\in B_r(v)} \deg(u).$
\end{lemma}}
\begin{proof}[Proof of Lemma~\ref{lm_localization_appendix}.]
	\ed{Without loss of generality we can assume that $r>e\cdot D$ since otherwise  the bound~\eqref{eq_localization_appendix} trivializes.}

	\ed{The realized subnetwork of $v$  is contained in the $(r-1)$-neighborhood if and only if no path connecting $v$ and the sphere  $S_r(v)=B_r(v)\setminus B_{r-1}(v)$ in the a-priori network $G$
	exists in the realized network $G_T$.
	Consider an agent $u\in S_r(v)$ and let $(v=u_0,u_1,u_2,\ldots,u_k=u)$  be a path connecting  $v$ and $u$ in $G$.  The chance that this path is presented in the realized network $G_T$ is bounded by the probability that $u_k,u_{k-1},\ldots, u_{1}$ arrive exactly in this order (we excluded the agent $u_{0}=v$  since we  condition on his arrival time); this probability is equal to $\frac{1}{k!}$. Each path between $v$ and the sphere $S_r(v)$ has length $k\geq r$.	The total number of different paths of length $k$ starting from $v$ is at most  $D^{k}$ since there are at most $D$ options for each of $k$ steps. }

\ed{Therefore, by the union bound, the chance that there is a path between $v$ and $S_r(v)$ in the realized network $G_T$ is at most
$$\sum_{k=r}^\infty \frac{D^k}{k!}.$$
Since $\frac{D^k}{k!}\leq \frac{D^r}{r!}\cdot \left(\frac{D}{r}\right)^{k-r}$, we can bound the sum by the geometric progression
$$\sum_{k=r}^\infty \frac{D^k}{k!}\leq \frac{D^r}{r!} \sum_{l=0}^\infty  \left(\frac{D}{r}\right)^{l} = \frac{D^r}{r!}\cdot \frac{1}{1-\frac{D}{r}}\leq \frac{1}{1-\frac{1}{e}}\cdot \frac{D^r}{r!}\leq 2\cdot  \frac{D^r}{r!},$$
where the last two inequalities follow from the assumption $r\geq e\cdot D$ and the fact that $\frac{1}{1-\frac{1}{e}}\leq 2$. 
Using the inequality $n!\geq \left(\frac{n}{e}\right)^n$, we can get rid of the factorial in the denominator: 
$$2\cdot \frac{\big(D\big)^{r}}{r!}\leq 2\cdot\left(\frac{e\cdot D}{r}\right)^{r}.$$
Hence, with probability at least  $1-2\cdot\left(\frac{e\cdot D}{r}\right)^{r}$ there are no paths between $v$ and $S_r(v)$ in $G_T$ or, equivalently, $G_{v,T}\subset B_{r-1}(v)$.}
\end{proof}

\section{Proof of Theorem~\ref{th_local_requirement} }\label{appendix_local_requirement}

\new{Denote  by $F_v^d=(u_i)_{i=1,\ldots, d}$ the collection of $v$'s friends from the local learning requirement and by $F_{v,T}^{d}=\{u\in F_v^d\, :\, t_u\leq t_v\}$  those of them who arrive earlier than $v$.} 
	
	Consider the following deviation $a_v'$ of  $v$ from his equilibrium action $a_v$. The agent decides to rely on his friends $F_{v,T}^{d}$: if the majority of them played action $a\in\{0,1\}$, the agent $v$ repeats this action; in the case of a tie or $v$ being isolated, $a_v=s_v$, i.e., $v$ follows his signal. In equilibrium, no deviation is profitable and, hence, $\P(a_v=\theta)\geq\P(a_v'=\theta).$
		The probability of $a_v'$ being a mistake can be estimated using the Hoeffding inequality (see Footnote~\ref{footnote_Hoeffding}).
	
	{To apply the Hoeffding inequality, we need independence of actions $(a_u)_{u\in F_{v,T}^{d}}$ and an upper bound on the probability that $a_u\ne \theta$. To make the actions independent, we need conditioning on the appropriate family of events. The following lemmas describe the desired family of events and establish the bound.}
	
	\ed{Recall that $B_{r, G\setminus v}(u)$ denotes the $r$-neighborhood of $u$ in the network obtained from $G$ by eliminating~$v$.	
		 Denote by $W_v$ the event that the realized subnetwork of each  $u\in F_{v,T}^{d}$ is contained in $B_{r-1, G\setminus v}(u)$, i.e., $W_v=\cap_{u\in F_{v,T}^{d}} \{G_{u,T}\subset B_{r-1, G\setminus v}(u)\}$.
	Note that, given $W_v$, the realized subnetworks of $u\in F_{v,T}^d$ are disjoint.}
				
	\begin{lemma}\label{lm_independence}
		{Fix $\theta_0\in\{0,1\}$, $t\in[0,1]$, and a subset of friends $F\subset F_{v}^d$ of agent $v$. The  actions $(a_u)_{u\in F}$ are independent conditional on  $\theta=\theta_0$, the arrival time  $t_v=t$ of $v$,  the set of friends who arrived earlier $F_{v,T}^d=F$, and $W_v$.}
	\end{lemma}
The second lemma ensures that $W_v$ is a high-probability event. \ed{It is a corollary of the \new{local nature of decisions property} (Lemma~\ref{lm_localization_appendix}).}
\begin{lemma}\label{lm_W_high_probability}
	The probability of $W_v$ conditional on arrival times of $v$ and all his friends from $F_{v}^d$	
	enjoys the following lower bound:
	$$\P\big(W_v\mid  t_v, \  (t_z)_{z\in F_{v}^d} \big)\geq 1-\psi,\ \
	\mbox{where} \ \ \psi=2d\cdot\left(\frac{e\cdot D}{r}\right)^{r}.$$
\end{lemma}
The third lemma provides an upper bound on the probability of the wrong action.
\begin{lemma}\label{lm_upper_bound_mistake}	
		{For any subset $F$ of $F_{v}^d$, an agent $u\in F$, and any state-symmetric equilibrium, the conditional probability of the wrong action satisfies the following inequality:
			\begin{equation}\label{eq_conditional_mistake_bound}
			\P(a_u\ne \theta \mid \theta =\theta_0, \  t_v=t, \  F_{v,T}^d=F, \ W_v)\leq \frac{1}{1-\psi}\left(1-p + \frac{3}{t\cdot \sqrt{\deg(u)-1}}\right),
			\end{equation}
			where $\psi$ is from Lemma~\ref{lm_W_high_probability}.
		}
	\end{lemma}

	{The lemmas are proved below. With their help, we complete the proof of Theorem~\ref{th_local_requirement}.
		By the Hoeffding inequality and Lemmas~\ref{lm_independence} and~\ref{lm_upper_bound_mistake}, the probability of the wrong action $a_v'$ for the high state ($\theta=1$) is bounded as follows:
		\begin{align}
		\P\big(a_{v}'\ne \theta\mid \theta=1, \  t_v, \  F_{v,T}, \ W_v\big) &\leq\P\left(\sum_{u\in F_{v,T}^{d}} a_u \leq  \frac{\big|F_{v,T}^{d}\big|}{2}\ \Big|\ \theta =1, \  t_v, \  F_{v,T}, \ W_v\right)\leq\\
		& \leq \exp\left(-2\left(\frac{1}{2}- \frac{1}{1-\psi}\left(1-p + \frac{3}{t_v\cdot \sqrt{d-1}}\right)\right)^2\cdot \big|F_{v,T}^{d}\big| \right)=
\label{eq_Hoeffding_proof_th_local}		
		\\
		&=\exp\left(-\frac{1}{2}\frac{\left(2p-1-\psi- \frac{6}{t_v\cdot \sqrt{d-1}}\right)^2}{(1-\psi)^2}\cdot \big|F_{v,T}^{d}\big| \right),\label{eq_Hoefding}
		\end{align}
		where we used that $\deg(u)\geq d$ for $u\in F_v^d$.		 The bound~\eqref{eq_Hoefding} holds only if $t_v$ is not too small: to apply the Hoeffding inequality, the expression in the parentheses,  $\left(2p-1-\psi- \frac{6}{t\cdot \sqrt{d-1}}\right)$, must be non-negative (see the requirement $x\geq 0$ in  Footnote~\ref{footnote_Hoeffding}). We will use the bound~\eqref{eq_Hoefding} only for $t_v$ such that this expression is at least $\frac{2p-1-\psi}{2}$ or equivalently $t_v\geq \frac{12}{\sqrt{d-1}(2p-1-\psi)}$. For small $t_v$, we roughly bound the probability by $1$ and get the following inequality valid for all $t_v$:
		\begin{equation}\label{eq_Hoefding_for_all_t}
		\P\big(a_{v}'\ne \theta\mid \theta=1, \  t_v, \  F_{v,T}, \ W_v\big)\leq \scalebox{1.3}{\mbox{$\mathds{1}$}}_{\left\{t_v< \frac{12}{\sqrt{d-1}(2p-1-\psi)}\right\}}+\exp\left(-\frac{1}{8}\frac{\left(2p-1-\psi\right)^2}{(1-\psi)^2}\cdot \big|F_{v,T}^{d}\big| \right),
		\end{equation}
		where $\mathds{1}_A$ denotes the indicator of an event $A$.
	}
	
	{Let us dispense with conditioning on $\theta=1, \  t_v, \  F_{v,T},$ and  $W_v$. For $\theta=0$, one similarly derives the same bound~\eqref{eq_Hoefding_for_all_t} and,  hence, it holds unconditionally for $\theta$.
		Since $\P(\cdot)\leq \P(\,\cdot\mid W_v)+(1-\P(W_v))$, we dispense with conditioning on $W_v$ at the cost of increasing the upper bound by $\psi$.
		It remains to average the bound over $F_{v,T}^{d}$ and $t_v$. The size $|F_{v,T}^{d}|$  is uniformly distributed over $\big\{0,1,\ldots,d\big\}$ because it is the length of the prefix of $v$ in a random permutation of $F_{v}^{d}\cup \{v\}$. Thus,  averaging the second summand in~\eqref{eq_Hoefding_for_all_t} results in the geometric progression of length $d$.  Taking into account that $t_v$ is uniformly distributed on $[0,1]$ and bounding the geometric progression of length $d$ by the infinite one, we get
		\begin{equation}\label{eq_Hoeffding_non_simplified}
		\P(a_v'\ne\theta)\leq \psi+ \frac{12}{\sqrt{d-1}(2p-1-\psi)}+\frac{1}{d+1}\cdot \frac{1}{1-\exp\left(-\frac{1}{8}\frac{(2p-1-\psi)^2}{(1-\psi)^2}\right)}.
		\end{equation}
		This upper bound can be simplified if we note that $1-\exp(-x)\geq x \cdot \frac{1-\exp(-a)}{a}$ for $x\in[0,a]$. Since $\frac{2p-1-\psi}{1-\psi}\leq 1$, we obtain $$\frac{1}{1-\exp\left(-\frac{1}{8}\frac{(2p-1-\psi)^2}{(1-\psi)^2}\right)}\leq \frac{8 (1-\psi)^2}{\left(1-\exp\left(-\frac{1}{8}\right)\right)(2p-1-\psi)^2}\leq \frac{72}{(2p-1-\psi)^2},$$
		where in the second inequality we bounded $(1-\psi)$ by $1$ and used that $1-\exp\left(-\frac{1}{8}\right)\geq \frac{1}{9}$. Denoting by $b$ and $c$ the second and the third summand in~\eqref{eq_Hoeffding_non_simplified}, respectively,   we see that $c\leq \frac{1}{2} b^2$. For $b$ below~$1$, we have $b^2\leq b$, while for $b\geq 1$ the bound~\eqref{eq_Hoeffding_non_simplified} becomes trivial anyway. Therefore,
		$$\P(a_v'\ne\theta)\leq \psi+ \frac{18}{\sqrt{d-1}(2p-1-\psi)}.$$
		Taking into account that $\P(a_v=\theta)\geq \P(a_v'=\theta)=1-\P(a_v'\ne\theta)$ and substituting the expression for $\psi$ from Lemma~\ref{lm_W_high_probability}, we complete the proof of Theorem~\ref{th_local_requirement}.\qed
	}
	% by the assumption $|F_v^{D}|\geq D$

\begin{proof}[Proof of Lemma~\ref{lm_independence}:]
To ensure conditional independence, we make use of the following property of the realized subnetwork 	\ed{$G_{u,T}=(V_{u,T},E_{u,T})$ of an agent $u\in V$.  Conditional on  $G_{u,T}$, his action $a_u$ is independent of $(s_z, a_z)_{z\in V\setminus V_{u,T}}$ and $(t_z)_{z\in V}$. In other words, $u$'s action is solely determined by his realized subnetwork and signals of agents from this subnetwork.}	
		This property has an important consequence. Consider a group of agents $U\subset V$ and an event determined by the collection of arrival times $\big\{T=(t_z)_{z\in V}\in \mathcal{T} \big\}$ for some $\mathcal{T}\subset [0,1]^V$. Then actions $(a_u)_{u\in U}$ are independent conditional on $\theta=\theta_0$ and  $\{T\in \mathcal{T}\}$ provided that all the realized subnetworks $(G_{u,T})_{u\in U}$ are independent disjoint random networks conditional on $\{T\in \mathcal{T}\}$.	
	
	With this general observation, we will demonstrate that actions $(a_u)_{u\in F}$ are conditionally independent given $ \theta =\theta_0, \  t_v=t, \  F_{v,T}^d=F$, and  $W_v$. 
	By the definition of $W_v$, the realized subnetwork of each $u_i\in F$ belongs to $B_{r,G\setminus v}(u_i)$ and thus the realized subnetworks  $(G_{u,T})_{u\in F}$ are disjoint. It remains to check their independence.
	
	\ed{Recall that the event $W_v$  has the form $\cap_{u\in F_{v,T}^d}\{G_{u,T}\subset B_{r-1,G\setminus v}(u)\}$. We now check that each event $\{G_{u,T}\subset B_{r-1,G\setminus v}(u)\}$ is determined by  arrival times $t_z$ of agents $z\in B_{r,G\setminus v}(u)\cup\{v\}$. Indeed, $G_{u,T}$ belongs to $B_{r-1,G\setminus v}(u)$ if and only if $u$ arrives earlier than $v$ and for any path  $(u=z_0,z_1,z_2,\ldots,z_{k})$ not passing through $v$ and connecting $u$ and the sphere $S_{r, G\setminus v}(u)= B_{r,G\setminus v}(u)\setminus B_{r-1,G\setminus v}(u)$ in the a-priori network $G$, 
	there is an index $i$ such that $t_{z_i}<t_{z_{i+1}}$. Consequently, we can rewrite the event  $\{G_{u,T}\subset B_{r-1,G\setminus v}(u)\}$ as $\{(t_z)_{z\in B_{r,G\setminus v}(u)}\subset \mathcal{T}_u(t_v)\}$, where $\mathcal{T}_u(t_v)$ is a certain subset of $[0,1]^{B_{r,G\setminus v}(u)}$ depending on the arrival time of $v$.}

{Conditioning on $F_{v,T}^d=F$, $t_v=t$,  and $W_v$ becomes equivalent to conditioning on the following family of events determined by arrival times:
 $\{(t_z)_{z\in B_{r,G\setminus v}(u)}\subset \mathcal{T}_u(t)\}$ for each $u\in F$, $t_u>t$ for $u\in F_v^d\setminus F$, and $t_v=t$. Arrival times $(t_z)_{z\in V}$ are unconditionally independent, and the condition restricts the values of disjoint subsets (here we use the fact that  $B_{r,G\setminus v}(u)$ are disjoint for $u\in F_v^d$). Thus, conditionally on $F_{v,T}^d=F$, $t_v=t$,  and $W_v$, the families of random variables $(t_z)_{z\in B_{r,G\setminus v}(u)}$ are independent across $u\in F$. 	The realized subnetwork $G_{u,T}$ of $u$ is determined by $(t_z)_{z\in B_{r,G\setminus v}(u)}$  and, hence, the networks $(G_{u,T})_{u\in F}$ are also conditionally independent, which implies the desired conditional independence of  actions $(a_u)_{u\in F}$.}
\end{proof}
\begin{proof}[Proof of Lemma~\ref{lm_W_high_probability}:]	
\ed{Recall that $W_v=\cap_{u\in F_{v,T}^d}\{G_{u,T}\subset B_{r-1,G\setminus v}(u)\}$ and consider one of these events. We pick an agent $u\in F_{v}^d$ with $t_u<t_v$ and demonstrate that
\begin{equation}\label{eq_proof_of_W_high_prob}
\P\big(G_{u,T}\subset B_{r-1,G\setminus v}(u)\mid t_v, (t_z)_{z\in F_v^d} \big)\geq 1-2\cdot \left(\frac{e\cdot D}{r}\right)^r.
\end{equation}
Note that the event $\{G_{u,T}\subset B_{r-1,G\setminus v}(u)\}$ is determined by the arrival times $t_z$ of $z\in B_{r, G\setminus v}(u)\cup \{v\}$ only (see the argument in the proof of Lemma~\ref{lm_independence}). Hence, by the independence of arrival times, we can simplify the condition 
$$\P\big(G_{u,T}\subset B_{r-1,G\setminus v}(u)\mid t_v, (t_z)_{z\in F_v^d} \big)= \P\big(G_{u,T}\subset B_{r-1,G\setminus v}(u)\mid t_v, t_u\big)$$
because $\big(F_v^d\setminus \{u\}\big)\cap B_{r, G\setminus v}(u)=\emptyset.$
Since we condition on $t_u$, the only information added by knowing $t_v$ is that \emph{at the time $u$ arrives, $v$ is absent} (recall that we assume that $t_u<t_v$). Thus
$$\P\big(G_{u,T}\subset B_{r-1,G\setminus v}(u)\mid t_v, t_u\big)= \P_{G\setminus v}\big(G_{u,T}\subset B_{r-1,G\setminus v}(u)\mid t_u\big),$$
where $\P_{G\setminus v}$ refers to the probability with respect to the arrival process in the network $G\setminus v$. By Lemma~\ref{lm_localization_appendix} applied to $u$ in the network  $G\setminus v$, we get
$$\P_{G\setminus v}\big(G_{u,T}\subset B_{r-1,G\setminus v}(u)\mid t_u\big)\geq 1-2\cdot \left(\frac{e\cdot D}{r}\right)^r$$
and we deduce~\eqref{eq_proof_of_W_high_prob}.
The desired inequality for the probability of $W_v$ follows from the union bound and~\eqref{eq_proof_of_W_high_prob}.}
\end{proof}
\begin{proof}[Proof of Lemma~\ref{lm_upper_bound_mistake}:]	
\ed{In the proof of Lemma~\ref{lm_independence} we observed that, once the realized  subnetwork $G_{u,T}$ of an agent  $u$ is given, his action $a_u$ is determined by the signals of agents from this subnetwork. We also saw that,
 conditional on $F_{v,T}^d=F$, $t_v=t$, and $W_v$, the realized subnetwork of $u\in F$ is determined by arrival times $(t_z)_{z\in B_{r,G\setminus v}(u)}$, and the condition  is equivalent to $\{(t_z)_{z\in B_{r,G\setminus v}(u)}\subset \mathcal{T}_u(t)\}$ (we use the notation introduced in that proof).  Therefore, the distribution of $G_{u,T}$ (and, hence, $a_u$) conditional on $\theta=\theta_0$, $F_{v,T}^d=F$, $t_v=t$, and $W_v$ is the same no matter what  other agents are  in $F$.}
%Now let's bound  the probability of a wrong action $a_u$ of an agent $u\in F$ conditional on our family of events: $\theta=\theta_0$, $F_{v,T}=F$, $t_v=t$, and $W_v$.
This observation allows us to simplify the condition:
	$$\P(a_u\ne \theta \mid \theta =\theta_0, \  t_v=t, \  F_{v,T}^d=F, \ W_v)= \P(a_u\ne \theta \mid \theta =\theta_0, \  t_u< t, \  v\notin F_{u,T}, \  W_v).$$
	By the state symmetry of the equilibrium, the latter probability does not change if we eliminate conditioning on $\theta=\theta_0$.
	This probability can be bounded as follows:
	\begin{equation}\label{eq_wrong_action_alphs_singled_out}
	\P(a_u\ne \theta \mid   t_u< t, \  v\notin F_{u,T}, \  W_v)\leq \frac{\P(a_u\ne \theta \mid \  t_u< t, \  v\notin F_{u,T})}{1-\psi}
	\end{equation}
by the formula of total probability and the lower bound $\P(W_v\mid t_u, t_v)\geq 1-\psi$.

{	It remains to estimate the numerator. We use the following observation: for any event $A$ that belongs to the information partition of agent $u$, the conditional probability of the wrong action $\P(a_u\ne \theta \mid A)$ is at most $1-p$. Otherwise, the agent can profitably deviate from his equilibrium strategy by following his signal whenever $A$ occurs.
}
	
	{The event $A'=\{t_u< t, \  v\notin F_{u,T}\}$ is not known to $u$ since $u$ does not observe his arrival time. However, we can approximate the event $A'$ by the event $A=\{\hat{t}_u< t, \  v\notin F_{u,T}\}$, where $\hat{t}_u=\frac{|F_{u,T}|}{\deg(u)-1}$ is a proxy for $u$'s arrival time (for large-degree agents $\hat{t}_u\approx t_u$ by the law of large numbers). Agent $u$ knows when $A$ occurs and, hence, $\P(a_u\ne \theta\mid A)\leq 1-p$. The conditional probability with respect to $A'$ can be bounded  as follows:
	\begin{align*}
	\P(a_u\ne \theta \mid A')&\leq \frac{\P(a_u\ne \theta , \ A)+ \P(A'\setminus A)}{\P(A')}=\\
	&= \P(a_u\ne \theta \mid A)\cdot \frac{\P(A)}{\P(A')}+ \frac{\P(A'\setminus A)}{\P(A')}\leq (1-p)\cdot \frac{\P(A)}{\P(A')}+ \frac{\P(A'\setminus A)}{\P(A')}.
	\end{align*}
	Let us estimate all the probabilities in this expression. The probability $\P(A')$ can be computed explicitly as
	$$\P(A')=\P(t_u<t,\ t_v>t_u)=\int_0^t dt_u\int_{t_u}^1 dt_v =\int_0^t (1-t_u) \,dt_u=t-\frac{t^2}{2}.$$
	To estimate $\P(A)$, we note that conditionally on $t_u$ and $t_v>t_u$, the number of friends observed by $u$ has the binomial distribution with parameters  $\deg(u)-1$ and $t_u$ (there are $\deg(u)-1$ friends and each of them arrives before time $t_u$ independently with probability $t_u$). By the Hoeffding inequality, we get
	\begin{align*}
	\P\left(\frac{|F_{u,T}|}{\deg(u)-1}\leq t \mid t_u, \  v\notin F_{u,T}\right)&\leq \exp\left(  -2 (t_u-t)^2(\deg(u)-1)\right) \ \ \mbox{for $t\leq t_u$}\\
	\P\left(\frac{|F_{u,T}|}{\deg(u)-1}\geq t \mid t_u, \  v\notin F_{u,T}\right)&\leq \exp\left(  -2 (t_u-t)^2(\deg(u)-1)\right) \ \ \mbox{for $t\geq t_u.$}
	\end{align*}
	Now we are ready to estimate $\P(A)$:
	\begin{align*}
	\P(A)&=\P\left(\frac{|F_{u,T}|}{\deg(u)-1}\leq t, \ t_v>t_u \right)=\int_{0}^1 dt_v\int_0^{t_v} \P\left(\frac{|F_{u,T}|}{\deg(u)-1}\leq t \mid t_u, \ t_v>t_u\right) \,dt_u \leq \\
	&\leq \int_{0}^t dt_v\int_0^{t_v} 1\,dt_u+  \int_{t}^1 dt_v\int_0^{t} 1\,dt_u +  \int_{t}^1 dt_v\int_t^{t_v} \exp\left(  -2 (t_u-t)^2(\deg(u)-1)\right)\,dt_u \leq\\
	&\leq \frac{t^2}{2}+t(1-t)+\int_{-\infty}^0 \exp\left(  -2 s^2(\deg(u)-1)\right)\,ds= t-\frac{t^2}{2}+\sqrt{\frac{\pi}{8(\deg(u)-1)}}.
	\end{align*}
	Similarly,
	\begin{align*}
	\P(A'\setminus A)&=\P\left( t_u<t,\ t_v>t_u, \ \frac{|F_{u,T}|}{\deg(u)-1}> t\right)=\\
	&=\int_0^t dt_u \left(\int_{t_u}^1\, dt_v\right) \cdot  \P\left(\frac{|F_{u,T}|}{\deg(u)-1}>t \mid t_u, \  v\notin F_{u,T}\right)\leq\\ 
	&\leq \int_{-\infty}^0 \exp\left(  -2 s^2(\deg(u)-1)\right)\,ds=\sqrt{\frac{\pi}{8(\deg(u)-1)}}.
	\end{align*}
	Putting all the pieces together, we obtain
	\begin{align*}
	\P(a_u\ne\theta\mid A')&\leq (1-p) \left(1+\frac{1}{t-\frac{t^2}{2}}\sqrt{\frac{\pi}{8(\deg(u)-1)}}\right)+\frac{1}{t-\frac{t^2}{2}}\sqrt{\frac{\pi}{8(\deg(u)-1)}}\leq\\
	&\leq 1-p + \frac{3}{t\cdot \sqrt{\deg(u)-1}}.
	\end{align*}
	In the last inequality we took into account that $t-\frac{t^2}{2}\geq \frac{t}{2}$, $p\geq0$, and $\sqrt{2\pi}\leq3$.}
	
	{Substituting this expression in~\eqref{eq_wrong_action_alphs_singled_out} leads to the desired bound~\eqref{eq_conditional_mistake_bound}.}
\end{proof}

\section{Missed proofs for Section~\ref{sec_robust}}\label{appendix_missed_proofs_robust}

\begin{proof}[Proof of Lemma~\ref{lm_robust}]
	%Now let's estimate this probability for a subnetwork $G^{V'}$. Note that the chance that friends observed by $v$ have no common predecessors can only increase after deleting some agents and/or edges from $G$. Hence, $\psi(G^{V'})\leq \psi(G)$ for any  $v\in V'$ and thus we can apply  Lemma~\ref{lm_no_predecessors} to those agents in $V'$ that have high enough degree
	\ed{Consider the induced subnetwork  $G^{V'}$ of $G$ for some $V'\subset V$ with $|V'|=\big\lceil \alpha\cdot|V|\big\rceil$ and 	
	denote by $\deg'(v)$ the degree of an agent $v\in V'$ in $G^{V'}$.} Fix positive  $\gamma\leq \alpha$. Our goal is to bound the fraction of agents that have $\deg'(v)<\gamma \cdot D=\gamma\cdot \deg(v)$. Denote by $V^1$ the set of all such agents $v\in V'$ and by $V^2$, the set of eliminated agents $V\setminus V'$.
	Apply inequality~\eqref{eq_mixing} (see mixing lemma~\ref{lm_mixing}) to these $V^1$ and $V^2$. Since $E(V^1,V^2)>(1-\gamma)D\cdot|V^1|$ and $(1-\alpha)|V|-1< |V^2|\leq (1-\alpha)|V|$, we get
	$$(1-\gamma)D\cdot|V^1|-\frac{D}{|V|}\cdot |V^1|\cdot (1-\alpha)|V|\leq |\lambda_2|\sqrt{|V^1|\cdot(1-\alpha)|V|}.$$
	Dividing both sides by $\sqrt{|V^1|}$ and rearranging the terms we get
	$$(\alpha-\gamma)D \sqrt{|V_1|}\leq |\lambda_2|\sqrt{(1-\alpha)|V|}$$
	and, therefore,
	$$|V^1|\leq \frac{(1-\alpha)|\lambda_2|^2}{(\alpha-\gamma)^2 D^2}|V|\leq \frac{(1-\alpha)|\lambda_2|^2}{\alpha(\alpha-\gamma)^2 D^2}|V'|.$$
	Thus at least a fraction $\left(1-\frac{(1-\alpha)|\lambda_2|^2}{\alpha(\alpha-\gamma)^2 D^2}\right)$ of agents  $v\in V'$ has $\deg(v)\geq \gamma\cdot D$.
	
	{By Lemma~\ref{lm_friends_with_high_degree} contained below, at least $\left(1-\frac{2(1-\alpha)|\lambda_2|^2}{\alpha(\alpha-\gamma)^2 D^2}\right)|V'|$ agents have at least $\frac{\gamma}{2}\cdot D$ friends with degree $\frac{\gamma}{2}\cdot D$ or higher.}

	Applying Theorem~\ref{th_local_requirement} combined with Lemma~\ref{lm_girth_and_llr} to each agent in this set and estimating the chance of the correct action outside this set by zero, we obtain the following bound on the learning quality for any state-symmetric equilibrium $\sigma'$ in $G^{V'}$:
	{$$L_{\sigma'}\left(G^{V'}\right)\geq \left(1-\frac{2(1-\alpha)|\lambda_2|^2}{\alpha(\alpha-\gamma)^2 D^2}\right)\left(1-\ed{\delta\left(p,\frac{\gamma}{2} D,r,D\right)}\right),$$
	where  \ed{$r=\Big\lfloor\frac{g-3}{2}\Big\rfloor$}.}

	{Taking into account that $\delta(p,D,r,D)>\frac{18}{\sqrt{D}}$, we see that the expression in the first parenthesis is greater than $1-\frac{(1-\alpha)|\lambda_2|^2}{9 \alpha(\alpha-\gamma)^2 D^{\frac{3}{2}}}\cdot \delta(p,D,r,D)$.
		It is easy to check that $\delta(p,\beta D,r,D)\leq \frac{1}{\sqrt{\beta-\frac{1}{D}}}\delta(p,D, r, D)$ for any $\beta$ and, hence, the expression in the second parenthesis is at least $1-\frac{1}{\sqrt{\frac{\gamma}{2}-\frac{1}{D}} }\cdot \delta(p, D,r, D)$.}
	
	{Opening the brackets and dispensing with positive terms, we get
		$$L_{\sigma'}\left(G^{V'}\right)\geq 1-\left(\frac{(1-\alpha)|\lambda_2|^2}{9 \alpha(\alpha-\gamma)^2 D^{\frac{3}{2}}}+\frac{1}{\sqrt{\frac{\gamma}{2}-\frac{1}{D}} }\right)\cdot  \delta(p,D, r, D).$$
		Picking $\gamma=\frac{2\alpha}{3}$ and assuming that $\frac{\alpha}{3}-\frac{1}{D}\geq \frac{\alpha}{4}$ we obtain the desired bound~\eqref{eq_average_quality_subexpander}. In the complementary case of small $\alpha<\frac{12}{D}$ the bound~\eqref{eq_average_quality_subexpander}  trivializes since $\frac{2}{\sqrt{\alpha}}\delta(p,D, r, D)\geq \frac{\sqrt{D}}{\sqrt{3}}\cdot \frac{18}{\sqrt{D}}>1$.}
\end{proof}

{\begin{lemma}\label{lm_friends_with_high_degree}
		If in a network $G=(V,E)$ at least $(1-\beta)|V|$ agents have degree $D$ or higher; then,  at least $\left(1-2{\beta}\right)|V|$ agents have at least $\frac{D}{2}$ friends with degree at least $\frac{D}{2}$.
	\end{lemma}	
	\begin{proof}
		Denote by $V_{<D}$ the set of agents with less than $D$ friends and by $V_{<\frac{D}{2},\geq \frac{D}{2}}$ the set of agents with less than $\frac{D}{2}$ friends with degree $\frac{D}{2}$ or higher. We know  $|V_{<D}|<\beta|V|$ and want to prove that $\big|V_{<\frac{D}{2},\geq \frac{D}{2}}\big|< 2{\beta}|V|$. Assume by way of  contradiction that $\big|V_{<\frac{D}{2},\geq \frac{D}{2}}\big|\geq 2{\beta}|V|$. Therefore, $V'=V_{<\frac{D}{2},\geq \frac{D}{2}}\cap (V\setminus V_{<D})$ contains at least $\beta|V|$ agents. Each $v\in V'$ has at least $D$ friends and at least $D-\frac{D}{2}=\frac{D}{2}$ of them have degree less than $\frac{D}{2}$; denote the set of such low-degree friends by $F_{v,<\frac{D}{2}}$. The set   $V_{<D}$ contains the union of $F_{v,<\frac{D}{2}}$ over $v\in V'$. We obtain
		$$\frac{2}{{D}} \sum _{v\in V' } \big|F_{v,<\frac{D}{2}}\big|\leq |V_{<D}|,$$
		where the factor $\frac{2}{D}$ originates because no $u\in F_{v,<D}$ is counted more than $\deg(u)<\frac{D}{2}$ times. Using the bounds on  cardinalities of all the sets in this expression, we conclude that
		$$\frac{2}{D}\cdot \beta|V|\cdot \frac{D}{2}< \beta|V|\Longleftrightarrow 1<1.$$
		This contradiction completes the proof.
	\end{proof}	
}

\section{Randomized robustness}\label{app_randomized_robustness}
\ed{In Section~\ref{sec_robust}, we demonstrated existence of a-priori networks satisfying a very strong notion of adversarial robustness: the network aggregates information even if  a subset of agents is eliminated in an adversarial way.
Here we consider a weaker notion of \emph{robustness to random elimination} and show that \emph{any network} has this property. Namely, for networks with high learning quality, even if a substantial fraction of agents leaves the network, the remaining agents find a way to learn the state even though most paths of information diffusion (those involving eliminated agents) disappear. 

We will demonstrate the randomized robustness under an additional assumption about the information available to agents. \new{Recall that $G_{v,T}=(V_{v,T}, E_{v,T})$ denotes the realized subnetwork of an agent $v$; see the definition in Section~\ref{subsect_local_nature}.}
\begin{assumption}\label{assump_extra_observation}
	Each agent $v$,  in addition to his signal and actions of friends who arrived earlier than him, observes the set of agents $V_{v,T}$ of his realized subnetwork (without their actions and arrival times).
	% (see Section~\ref{subsect_local_nature} for the definition of the realized subnetwork of an agent). 
	 In other words, an agent knows the set of those who can possibly affect his action. The information set of $v$ is therefore $I_v=(s_v, (a_u)_{u\in F_{v,T}},V_{v,T})$.
\end{assumption}
\begin{remark}[The role of Assumption~\ref{assump_extra_observation}] We believe that this assumption plays a technical role and the randomized robustness must be ubiquitous without it as well; however, we were unable to get rid of it in our proof.
Assumption~\ref{assump_extra_observation} ensures the sequential structure of equilibrium. If $V_{v,T}=\emptyset$,  agent $v$ follows his signal. If $V_{v,T}$ is non-empty with $|V_{v,T}|=k$, the equilibrium strategy $\sigma_v\big(s_v, (a_u)_{u\in F_{v,T}}, V_{v,T}\big)$ is the optimal reply to $(\sigma_u)_{u\in V_{v,T}}$ with $|V_{u,T}|\leq k-1$ (since $V_{u,T}$ is a strict subset of $V_{v,T}$). This sequential structure implies that an agent gains no advantage from learning the set of agents who have not yet arrived, the property critical in the proof of Theorem~\ref{th_random_elimination}, the main result of this section.
\end{remark}

For a given  network $G$ and the probability $p$ of the correct signal, denote by  $L(G)$ the learning quality for the best equilibrium: $L(G)=\max_\sigma L_\sigma(G)$.} 
\begin{theorem}[Learning is robust to random elimination]\label{th_random_elimination}
	Under the Assumption~\ref{assump_extra_observation}, consider a network $G=(V,E)$ that has learning quality $L(G)= 1-\delta$ with some $\delta> 0$. 
	
	Fix $\alpha\in(0,1)$ and pick a subset $V'\subset V$ with $\big\lceil\alpha\cdot |V|\big\rceil$ agents %such that $p\cdot |V|\leq |V'|< p\cdot |V|+1$
	uniformly at random. Then the  learning quality for the induced subnetwork $G^{V'}$ enjoys the lower bound $L\big(G^{V'}\big)\geq 1-\sqrt{\frac{\delta}{\alpha}}$ with probability at least $1-\sqrt{\frac{\delta}{\alpha}}$ with respect to the choice of $V'$.
	%\fed{Rann, I did not mention here "induced equilibrium" because otherwise we should explain what it is...}
\end{theorem}
%\begin{corollary} Consider a sequence of of networks $G_n$ that admits asymptotic learning. Eliminate randomly {$(1-\alpha)$-fraction of all agents from each network independently across the sequence ($\alpha$ can be arbitrary close to zero).
%		By Theorem~\ref{th_random_elimination}, with probability~$1$, the resulting  random sequence of networks $G'_{n}$ contains a subsequence $G'_{n_m}$ that admits asymptotic learning; moreover, one can assume $\lim_{m\to\infty}\frac{n_m}{m}=1$, i.e., only a negligible fraction of networks must be thrown away.}
%	
%Indeed, by Theorem~\ref{th_random_elimination},  {$L(G'_n)\geq 1-\sqrt{\frac{\delta_n}{\alpha}}$ with probability $1-\sqrt{\frac{\delta_n}{\alpha}}$, where  $\delta_n=1-L(G_n)$.} Deleting all networks $G'_n$ where   {$L(G'_n)<1-\sqrt{\frac{\delta_n}{\alpha}}$,} we get a sequence $G'_{n_m}$. By the assumption $\delta_n$ (and hence $\delta_{n_m}$) tends to zero and thus the sequence $G'_{n_m}$ admits asymptotic learning. The condition $\lim_{m\to\infty}\frac{n_m}{m}=1$ follows from the strong law of large numbers.
%\end{corollary}
%Here we present the proof sketch for Theorem~\ref{th_random_elimination}; the formal proof can be found in Appendix~\ref{appendix_missed_proofs_random}.
\begin{proof}[Proof of Theorem~\ref{th_random_elimination}.]
	The argument is based on a coupling of the learning process in the original network $G$ and the selection of the random subnetwork $G^{V'}$. Fix an equilibrium $\sigma=(\sigma_v)_{v\in V}$  maximizing $L_\sigma(G)$ and pick a subset $V'$ to be the set of $\big\lceil \alpha\cdot |V|\big\rceil$ earliest arrivals.
	
	For $v\in V'$, the equilibrium strategy $\sigma_v$ in the original network $G$ can be used as a strategy in $G^{V'}$. The resulting family of strategies $(\sigma_v)_{v\in V'}$ constitutes an equilibrium in $G^{V'}$, which we denote by $\sigma^{V'}$. Note that here we use Assumption\footnote{In the game played over $V'$, all the present agents know that agents from $V\setminus V'$ are absent. In particular, without Assumption~\ref{assump_extra_observation} an agent $v$ in $V'$ would know more about his predecessors than the same agent in the game played over $V$. Thus, a best reply in $V$ may no longer be a best reply in the game restricted to $V'$, even if all other agents maintained their strategy.}~\ref{assump_extra_observation}.
	
	%			 If $v$ had a profitable deviation $\sigma'_v$ in $G^{V'}$, he could also improve his strategy $\sigma_v$ in $G$ by playing $\sigma'_v$ whenever both $v$ and the set of agents $V_{v,T}$ from his realized subnetwork belong to $V'$. \\
	%			\ \ \phantom{AB} Without $(\bigstar)$,  $v$ could deduce some information about $V_{v,T}$ from knowing that the set of agents is $V'$ instead of $V$  (e.g., agents from  $V\setminus V'$ are not in $V_{v,T}$) and adjust his behavior accordingly. %Observability of $V_{v}$ implies
	%			%	that $v\in V'$ with $V_v\subset V'$ gets no extra information from knowing whether the network  is $G$ or $G^{V'}$.
	
	{The constructed coupling allows us to link the learning quality for $G$ and for $G^{V'}$ under equilibria $\sigma$ and $\sigma^{V'}$, respectively.
		By the formula of total probability, the learning quality for $G$ can be represented as
		$$L(G)=\frac{1}{|V|}\sum_{v\in V} \Big(\P(a_v=\theta\mid v\in V')\cdot \P(v\in V')+ \P(a_v=\theta\mid v\notin V')\cdot \P(v\notin V')\Big).$$
		Using a rough estimate $\P(a_v=\theta\mid v\notin V')\leq 1$ on the probability of the  correct action outside $V'$ and taking into account that $\P(v\in V')$  is bounded from below by $\alpha$, we get the following inequality:
		$$L_\sigma(G)\leq \alpha\cdot  \E_{V'} L_{\sigma^{V'}}\big(G^{V'}\big)+(1-\alpha),$$
		where $\E_{V'}$ denotes expectation with respect to the choice of $V'$. Since the left-hand side is equal to  $1-\delta$, we obtain
		$$1-\E_{V'} L_{\sigma^{V'}}\big(G^{V'}\big) \leq \frac{\delta}{\alpha}.$$
		Application of the Markov inequality completes the proof.}
	%If a substantial fraction of networks $G^{V'}$ constructed this way {have low learning quality,} then early arrivers in the original network $G$ make a substantial number of mistakes. Since the learning quality  for $G$ is at least $1-\delta$, most of networks  $G^{V'}$ must also have high enough learning quality to be compatible with this bound.
\end{proof}

\end{document}